\documentclass[11pt]{article}
\usepackage{color}
\usepackage[normalem]{ulem}
\usepackage{tikz}
\usepackage{todonotes}
\usepackage{epsfig,enumerate,float,amsmath,amsfonts,amsbsy,amssymb,amsthm,mathrsfs,lineno,ifpdf,relsize,float,etoolbox}
\usepackage{xcolor}
\usepackage[colorlinks=true, allcolors=blue]{hyperref}
\usepackage[ruled,vlined,linesnumbered,procnumbered]{algorithm2e}
\usepackage[mathscr]{euscript}
\usepackage{graphicx}
\usepackage{setspace}
\usepackage{makecell}
\usepackage[utf8]{inputenc}
\usepackage[english]{babel}
\usepackage[numbers]{natbib}
\usepackage[usenames,dvipsnames]{pstricks}
\usepackage{pst-grad}
\usepackage{float}
\usepackage{tikz}
\usetikzlibrary{positioning,fit,backgrounds,arrows.meta}
\theoremstyle{plain}
\newtheorem{theorem}{Theorem}
\newtheorem{lemma}{Lemma}
\newtheorem{coro}{Corollary}
\newtheorem{defi}{Definition}
\newtheorem{obs}{Observation}
\newtheorem{claim}{Claim}

\makeatletter
\AtBeginEnvironment{procedure}{\let\c@algocf\c@procedure}
\makeatother

\let\oldenumerate\enumerate
\renewcommand{\enumerate}{
	\oldenumerate
	\setlength{\itemsep}{1.5pt}
	\setlength{\parskip}{0pt}
	\setlength{\parsep}{0pt}
}
\title{Complexity and algorithms for proper conflict-free coloring in graphs}
\author{$^1$Dinabandhu Pradhan, $^1$Vaishali Sharma, \\ \\
	$^{1}$Department of Mathematics \& Computing\\ Indian Institute of Technology (ISM), Dhanbad\\
	\small \tt Email: dina@iitism.ac.in; vaishali.sharma7791@gmail.com}

\begin{document}

	\maketitle
	\begin{abstract}
		A \emph{proper conflict-free (PCF) $k$-coloring} of a graph $G$ is a proper $k$-coloring such that there exists a color that appears exactly once in the neighborhood of every non-isolated vertex $v\in V(G)$.
		The \emph{PCF chromatic number}, denoted by $\chi_{pcf}(G)$, is the least integer $k$ such that there exists a PCF $k$-coloring of $G$. Given a graph $G$ and a positive integer $k$, PCF $k$-\textsc{colorability} is to decide whether $G$ admits a PCF $k$-coloring. Ahn et al. [Discrete Appl. Math. 377 (2025) 10-17] proved that PCF $k$-\textsc{colorability} is NP-complete for bipartite graphs. We strengthen this result by proving that PCF $k$-\textsc{colorability} is NP-complete for perfect elimination bipartite graphs, which is a proper subclass of bipartite graphs. We also show that the PCF chromatic number of a graph cannot be approximated within $O(n^{1-\varepsilon})$ unless P=NP, for any $\varepsilon>0$. 
		
		\medskip
		On the positive side, we provide linear-time algorithms for PCF $k$-\textsc{colorability} in block graphs, proper interval graphs, chain graphs, and pseudo-split graphs. We show that $\chi_{pcf}(G)\leq \omega(G)+1$ for block graphs, proper interval graphs, and pseudo-split graphs (except $C_5$), and we characterize all graphs for which the equality holds.
		
	\end{abstract}
	
	\bigskip
	\noindent\textbf{Keywords:}
	Coloring; PCF coloring; block graphs; proper interval graphs;  algorithms; NP-completeness.

	\section{Introduction and notation}
	
	All graphs considered in this paper are finite, simple, connected, and undirected. We follow the notation and terminology used in \cite{textBook}. Let $G=(V,E)$ be a graph. The neighborhood of a vertex $v\in V(G)$ is $N_G(v)=\{u\in V(G): uv\in E(G)\}$ and the closed neighborhood of $v$ is $N_G[v]=N_G(v)\cup\{v\}$. When the context of the graph is clear, we simply write $N(v)$ and $N[v]$. For $k\in \mathbb{N}$, we use the notation $[k]$ to denote the set $\{1,2,\ldots,k\}$. For a set $S\subseteq V(G)$, $G[S]$ denotes the subgraph induced by vertices of $S$. 
	A \emph{proper $k$-coloring} of a graph $G$ is an assignment $c:V(G) \to [k]$ such that $c(u) \neq c(v)$ for each $uv \in E(G)$. The \emph{chromatic number} of a graph $G$, denoted by $\chi(G)$, is the minimum $k$ such that there exists a proper $k$-coloring of $G$.
	
	\medskip
	
	A hypergraph $\mathcal{H}=(\mathcal {V(H), E(H)})$ consists of a vertex set $\mathcal{V(H)}$ and an edge set $\mathcal{E(H)}$ of (hyper)edges such that each $e\in \mathcal{E(H)}$ is a subset of the vertex set $\mathcal{V(H)}$ of arbitrary size. 
	Even et al.~\cite{ELRS} introduced the notion of conflict-free coloring for hypergraphs. A conflict-free coloring of a hypergraph $\mathcal{H}$ is an assignment of $k$ colors to the vertices of the hypergraph $\mathcal{H}$ such that for each (hyper)edge $e \in \mathcal{E(H)}$ there exists a color that appears exactly once on its vertices. The study of conflict-free coloring originated from frequency-assignment and sensor network problems. Cheilaris~\cite{C} introduced the concept of conflict-free coloring for graphs. A conflict-free coloring of a graph $G$ is an assignment of $k$ colors to the vertices of the graph $G$ such that every non-isolated vertex $v\in V(G)$ has a color that appears exactly once on its neighborhood. 
	
%	Several variants of this coloring are known and have been studied previously \cite{AADFGHKS, BKM1, BKM2, BKM3, BKP, CKP, CL, HGY}.

	\medskip
	
	A \emph{proper conflict-free $k$-coloring} (or \emph{PCF $k$-coloring}, for short) of a graph $G$ is a proper $k$-coloring such that there exists a color that appears exactly once in the open neighborhood of every non-isolated vertex $v\in V(G)$.
	The \emph{PCF chromatic number} of a graph $G$, denoted by $\chi_{pcf}(G)$, is the least integer $k$ for which $G$ admits a PCF $k$-coloring.
	The proper conflict-free coloring was formalized by Fabrici et al.~\cite{FLRS} in their study of unique-maximum neighborhood colorings. They proved that every planar graph admits a PCF $8$-coloring, and constructed planar graphs requiring at least $6$ colors. Cho et al. \cite{CCKP} completely determined the threshold on the maximum average degree of a graph G, denoted by $\text{mad}(G)$, that guarantees a PCF $c$-coloring
	for all $c$ and also provided tightness examples. In addition, they proved that every planar graph with girth at least $5$ admits a proper conflict-free $7$-coloring. Caro et al. \cite{CPS} studied the PCF coloring for several basic graph classes including trees, cycles, hypercubes, and subdivisions of complete graphs. Liu \cite{Liu} showed that $\chi_{pcf}$ is bounded in several graph classes, including odd-minor–free graphs and graphs of bounded layered treewidth.
	
	% \medskip
	% An \emph{odd coloring} of a graph $G$ is a proper coloring such that there exists a color appearing an odd number of times in the neighborhood of every non-isolated vertex $v\in V(G)$.
	% The \emph{odd chromatic number}, denoted by $\chi_{o}(G)$, is the least integer $k$ for which $G$ admits an odd $k$-coloring of.
	% Every PCF-coloring is an odd coloring, and hence $\chi_o(G) \le \chi_{pcf}(G)$. However, $\chi_o(G)$ and $\chi_{pcf}(G)$ can differ significantly.
	% The notion of odd coloring was introduced by Petruševski and Škrekovski \cite{PS1}, and subsequently studied (see \cite{KO,Priyamvada}).
	%\medskip
	% Bhyravarapu et al. \cite{BKR} studied the parameterized complexity of \textsc{Odd $k$-coloring}. They have shown that \textsc{Odd $k$-coloring} can be solved in polynomial time on cographs and split graphs but remains NP-complete on perfect elimination bipartite graphs and star-convex bipartite graphs. 
	
	\medskip
	
	Given a graph $G$, the decision version of finding a PCF coloring of $G$ is defined as follows:
	
	\begin{center}
		\fbox{
			\begin{minipage}{0.9\linewidth}
				\underline{PCF $k$-\textsc{colorability}}
				
				\vspace{0.5em}
				
				\textbf{Input:} A graph $G = (V, E)$ and an integer $k>0$.
				
				\textbf{Question:} Does there exist a PCF coloring 
				$f : V(G) \to [k]$ of $G$?
			\end{minipage}
		}
	\end{center}

	Caro et al. \cite{CPS} proved that \textsc{PCF $k$-colorability} is NP-complete for general graphs. 
	Ahn et al. \cite{AIO} proved that PCF $k$-\textsc{colorability} is NP-complete for bipartite graphs. Recently, Sharma et al.~\cite{SPP} showed that the PCF chromatic number is bounded by a constant for AT-free graphs, split graphs, and $P_4$-sparse graphs. They also improved existing upper bounds for minor-closed $k$-planar graphs and $K_t$-minor-free graphs.
	To the best of our knowledge, no other algorithmic results are known for PCF $k$-\textsc{colorability}. 
	In this paper, we analyze the computational complexity of \textsc{PCF $k$-coloring} in several graph classes. The contributions of this paper are summarized as follows.
	
	\begin{itemize}
		\item We prove that \textsc{PCF $k$-colorability} is NP-complete for perfect elimination bipartite graphs.
		\item We show that the PCF chromatic number of a graph cannot be approximated within $O(n^{1-\varepsilon})$ for any $\varepsilon>0$, unless P=NP. The reduction used to prove the hardness result for approximating the PCF chromatic number yields that PCF $k$-\textsc{colorability} is NP-complete for dually chordal graphs.
		\item We present linear-time algorithms to compute an optimal PCF coloring for block graphs, proper interval graphs, chain graphs, and pseudo-split graphs. Note that block graphs and proper interval graphs are dually chordal graphs, whereas chain graphs are perfect elimination bipartite graphs.
		
		\item We show that $\chi_{pcf}(G)\leq \omega(G)+1$ for block graphs, proper interval graphs, and pseudo-split graphs (except $C_5$), and we characterize all graphs for which the equality holds.
	\end{itemize}
	
	\section{NP-completeness}
	
	PCF $k$-\textsc{colorability} is shown to be NP-complete for bipartite graphs \cite{AIO}. In this section, we strengthen this result by showing that PCF $k$-\textsc{colorability} is NP-complete for perfect elimination bipartite graphs.
	
	\medskip
	
	Let $G=(X,Y,E)$ be a bipartite graph. An edge $xy\in E(G)$ is called \emph{bisimplicial} if the subgraph induced by $N(x)\cup N(y)$ is a complete bipartite graph. A \emph{perfect edge elimination scheme} of a graph $G$ is an ordering of edges $(x_1y_1,\ldots,x_ky_k)$ such that each edge $x_{j+1}y_{j+1}$ is bisimplicial in the graph $G[V(G)\setminus\{x_1, y_1, x_2, y_2, \ldots, x_{j}, y_{j}\}]$.
	A graph is a perfect elimination bipartite graph if it admits a perfect edge elimination scheme.
	
	\medskip

	% In this section, we establish the NP-completeness of PCF $k$-\textsc{colorability} for $k \geq 3$ in perfect elimination bipartite graph. 
	First, we show that PCF $3$-\textsc{colorability} is NP-complete for perfect elimination bipartite graphs (see Theorem~\ref{np-pebg-thm1}). Then we show that for each $k\geq 4$, PCF $k$-\textsc{colorability} is NP-complete for perfect elimination bipartite graphs (see Theorem~\ref{np-pebg-thm2}). 
	
	\begin{theorem}\label{np-pebg-thm1}
		PCF $3$-\textsc{colorability} is NP-complete for perfect elimination bipartite graphs.
	\end{theorem}
	\begin{proof}
		We reduce an instance of $3$-\textsc{colorability}  into an instance of PCF $3$-\textsc{colorability}  for perfect elimination bipartite graphs. Let $G$ be a graph with the vertex set $V(G) = \{v_1, v_2, \ldots, v_n\}$. We use the same construction used in the proof of Theorem $2.2$ of~\cite{XXWH}. For the completeness of the paper, we provide the construction of the graph $G'$ from $G$.
		
		\begin{figure}[h]
			\begin{center}
				\includegraphics[width=\linewidth]{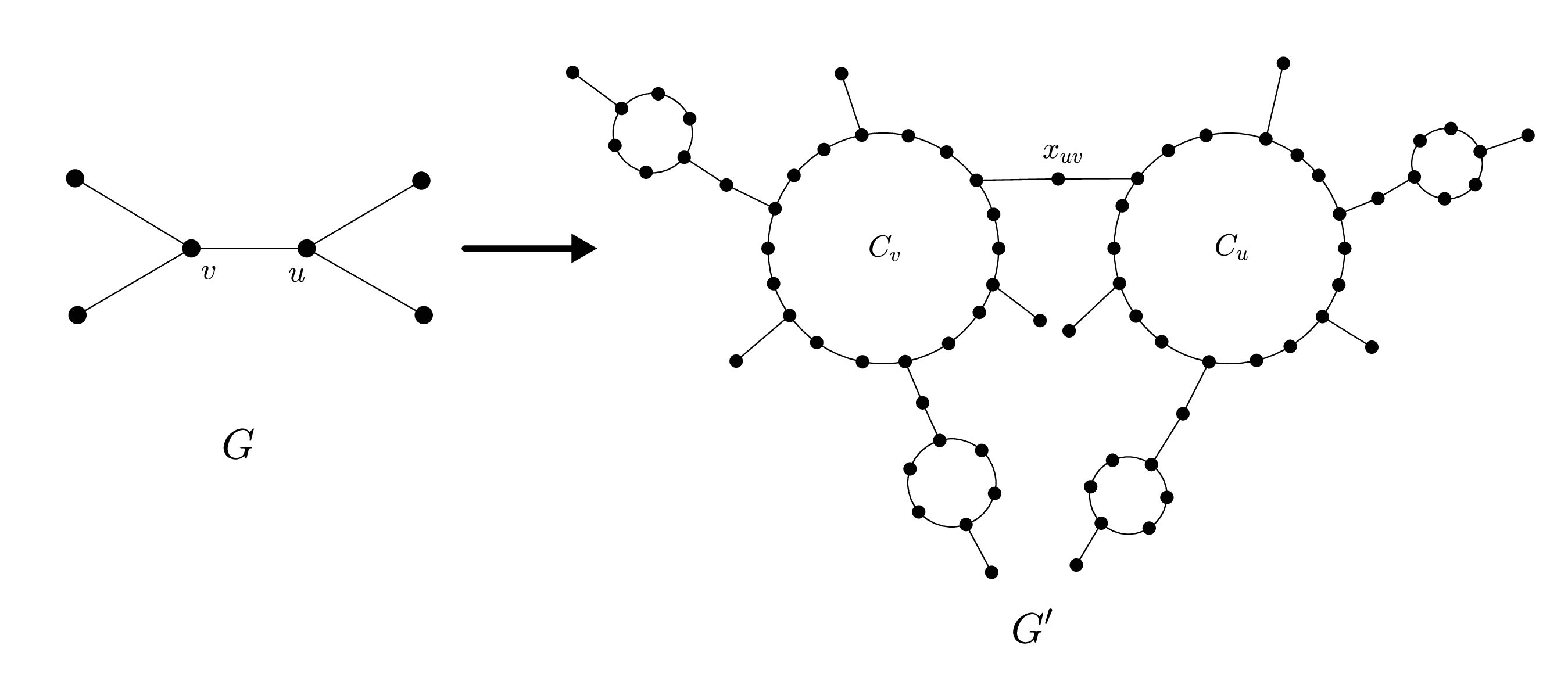}
				\caption{Construction of the graph $G'$ from the graph $G$.}
				\label{k=3}
			\end{center}	
		\end{figure}

		\begin{enumerate}
			\item Replace each vertex $v_i$ of $G$ by a cycle $C_{v_i}$ of length $6\cdot\deg(v_i)$, with an outgoing ‘‘half-edge’’ at each of the positions
			$1, 4, 7, 10, 13, \ldots$ along the cycle, representing twice the $\deg(v)$ edges in $G$ incident with $v$, and all other vertices on $C_v$ having
			degree $2$. We denote every vertex of $C_{v_i}$ by $v_{i_1}, v_{i_2}, \ldots, v_{i_{6\cdot\deg(v)}}$ and every vertex at the half edge of the positions $4, 10, 16, \ldots$ by $v_{i_4}', v_{i_{10}}', v_{i_{16}}', \ldots$.
			\item For every edge $uv \in E(G)$ we identify one of the half-edges at positions $1, 7, 13, \ldots$ from each of the $C_u$ and $C_v$,
			and we glue the two half-edges together at a newly added vertex $x_{uv}$ with degree $2$.
		\end{enumerate}
		
		\begin{claim} \label{reduction3}
			$G'$ is a perfect elimination bipartite graph.
		\end{claim}
		\begin{proof}
			% First, we show that $G'$ is a bipartite graph. 
			The new graph $G'$ is clearly bipartite, since every edge of $G$ is replaced by a $P_2$ while the segments of the cycles representing the vertices of $G$ are of even length equal to $6$. 
			Now, we provide a perfect edge elimination scheme for the bipartite graph $G'$. Let $E_i = (v_{i_4}'v_{i_4}, v_{i_3}v_{i_2}, v_{i_5}v_{i_6}, v_{i_{10}}'v_{i_{10}}, v_{i_9}v_{i_8}, v_{i_{11}}\\v_{i_{12}}, \ldots, v_{i_{6\cdot\deg(v_i)-2}}'v_{i_{6\cdot\deg(v_i)-2}}, v_{i_{6\cdot\deg(v_i)-3}}v_{i_{6\cdot\deg(v_i)-4}}, v_{i_{6\cdot\deg(v_i)-1}}v_{i_{6\cdot\deg(v_i)}})$ for each $i\in[n]$. 
			Let $X_{u_iv_i} = (x^{1}_{u_iv_i}x^{2}_{u_iv_i} )$ such that $x^{1}_{u_iv_i}x^{2}_{u_iv_i}\in E(G')$ is one of the edges corresponding to the edge $u_iv_i\in E(G)$. 
			Let $\sigma = (E_1, E_2, \ldots, E_n, X_{e_1}, X_{e_2}, \ldots, X_{e_m}) = \{e_1', e_2', \ldots, e_k'\}$. Let $S_i=\{e_1', e_2', \ldots, e_{i-1}'\}$. Let $S_i'$ be the set of endpoints of each edge $e_i'\in S_i$.
			Since every edge $e_i'\in\sigma$ is bisimplicial in $G'[V(G)\setminus S_i']$ for each $i \in [k]$. Also, the graph $G'[V(G') \setminus S_k']$ does not have any edge. Therefore, $\sigma$ is a perfect edge elimination scheme of the graph $G'$. Thus $G'$ is a perfect elimination bipartite graph.
		\end{proof}
		
		It has been shown in Theorem $2.2$ of \cite{XXWH} that 
		$G$ is $3$-colorable if and only if $G'$ is PCF $3$-colorable. Therefore, PCF $3$-\textsc{colorability} is NP-complete for perfect elimination bipartite graphs. 
	\end{proof}
	
	\begin{theorem}\label{np-pebg-thm2}
		For $k\geq 4$,   PCF $k$-\textsc{colorability} is NP-complete for perfect elimination bipartite graphs. 
	\end{theorem}
	\begin{proof}
		We reduce an instance of $k$-\textsc{colorability} in general graphs to an instance of PCF $(k+1)$-\textsc{colorability} in perfect elimination bipartite graphs. 
		Let $G$ be a graph with $V(G) = \{v_1, v_2, \ldots, v_n\}$ and $E(G) = \{e_1, e_2, \ldots, e_m\}$. Let $G'$ be constructed from $G$ as follows.
		
		\begin{enumerate}
			\item Subdivide each edge $e_i\in E(G)$ and denote the new vertex by $v_{e_i}$. Let $V_E = \{v_{e_1}, v_{e_2}, \ldots, v_{e_m}\}$.
			\item Attach a pendant vertex $v_i'$ to every vertex $v_i \in V(G)$. Let $V_1 = \{v_1', v_2', \ldots, v_n'\}$.
			\item For each $i\in [n]$, take a copy $x_i$ of every vertex $v_i \in V(G)$ and make it adjacent to its original vertex, that is, $x_iv_i \in E(G')$. Let $V_2 = \{x_1, x_2, \ldots, x_n\}$. 
			\item Take a $k$-clique $A$, where  $A = \{a_1, a_2, \ldots, a_k\}$. Connect every vertex $x_i$ of $V_2$ to every vertex $a_j$ of the clique $A$ by a common $2$-neighbor $z^i_j$. Let $V_3 = \{z^i_j\colon ~ i\in [n], j\in[k]\}$.
			\item Attach a pendant vertex $a_i'$ to every vertex of the clique $A$. Let $V_4 = \{a_1', \ldots, a_k'\}$.
			\item Subdivide every edge of the clique $A$. Represent every new vertex corresponding to the edge $a_ia_j\in A$ by $a_{ij}$. Let $V_5 = \{a_{ij}: i, j\in[k],  i\neq j\}$.
		\end{enumerate}
		
		\begin{figure}[h]
			\begin{center}
				\includegraphics[width=\linewidth, height=0.5\textheight]{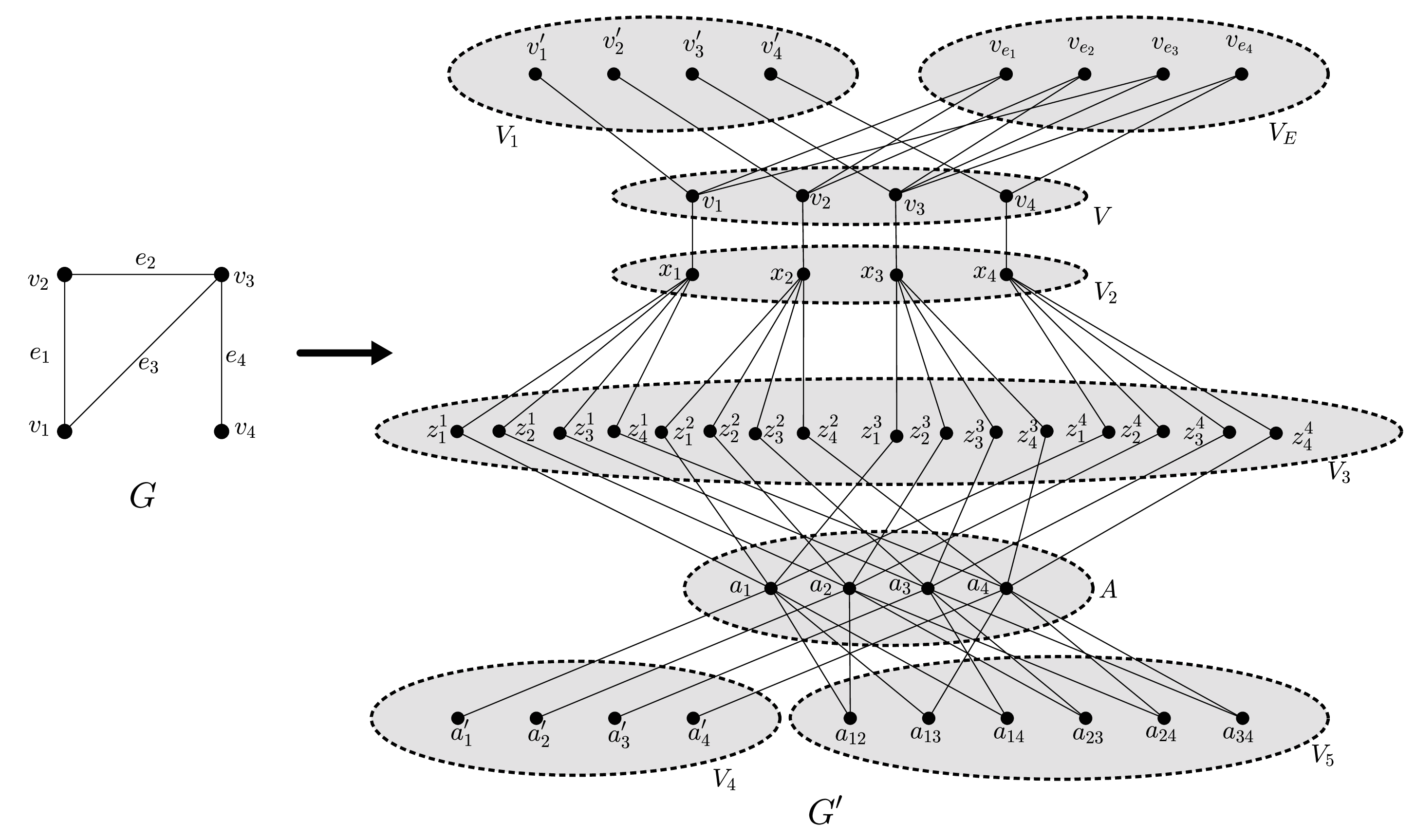}
				\caption{Construction of the graph $G'$ from the graph $G$ for the case $k=3$.}
				\label{k=3}
			\end{center}	
		\end{figure}
		
		\begin{claim} \label{reduction_gen}
			$G'$ is a perfect elimination bipartite graph.
		\end{claim}
		\begin{proof}
			First, we prove that $G'$ is a bipartite graph.
			Let $X = V \cup V_3 \cup V_4 \cup V_5$ and $Y = V_E \cup V_1 \cup V_2 \cup A$. Clearly, $X$ and $Y$ are independent sets. Therefore, $G'$ is a bipartite graph. Now we provide a perfect edge elimination scheme for the bipartite graph $G'$. Let $\sigma = (e_1' = a_1' a_1, e_2' = a_2' a_2, \ldots, a_k' a_k, z^1_1 x_{1}, z^2_1 x_{2}, \ldots, z^n_1 x_{n}, v_1' v_1, e_{l-1}' = v_2' v_2, \ldots, e_l' = v_n' v_n)$. For every edge $e_i'\in \sigma$, let $S_i = \{uv\colon~ uv\in\{e_1', e_2', \ldots, e_{i-1}'\} \}$. 
			Since every edge $e_i'\in \sigma$ is a bisimplicial edge in the graph $G'[V(G') \setminus S_i]$ for each $i$. Moreover, there are no edges in the graph $G'[V(G') \setminus S_l]$. Therefore, $\sigma$ is a perfect edge elimination scheme of the graph $G'$. Thus, the constructed graph $G'$ is a perfect elimination bipartite graph.
		\end{proof}
		
		\begin{claim} \label{(k)PCF_PEBG}
			$G$ is $k$-colorable if and only if $G'$ is PCF $(k+1)$-colorable for $k \geq 3$.
		\end{claim}
		\begin{proof}
			Let $G$ be $k$-colorable, $k\geq 3$. Let $f$ be a proper $k$-coloring of the graph $G$ with colors $1, 2, \ldots, k$. Let $f'$ be a coloring of the graph $G'$ such that $f'(v_i) = f(v_i)$ for each vertex $v_i \in V$, $f'(v_i') = c \in \{1, 2, \ldots, k\}\setminus f(v_i)$ for each vertex $v_i'\in V_1$, $f'(x_i) = k+1$ for each $x_i\in V_2$, $f'(a_i) = i$ for each $a_i\in A$, $f'(a_i') = k+1$ for each $a_i' \in [k]$, $f'(a_{ij}) = c \in \{1, 2, \ldots, k\}\setminus\{i, j\}$ for each $a_{ij}\in V_5$. Assign every vertex $v$ of the set $V_E$ a color from the set $\{1, 2, \ldots, k\}$ by avoiding the colors of its neighbors. Assign every vertex $z^i_j\in V_3$ a color from the set $\{1, 2, \ldots, k\}\setminus\{j\}$.
			Clearly, adjacent vertices of $G'$ are assigned different colors under $f'$. Thus, $f'$ is a proper coloring of the graph $G'$. 
			Clearly, the color $k+1$ appears exactly once in $N(v_i)$ for each vertex $v_i\in V$. Thus, each vertex $v_i\in V$ has a unique color in $N(v_i)$. The neighbors $s, t \in V$ of the vertex $v\in V_E$ are adjacent in $G$; therefore assigned different colors under $f$ and $f'$. Thus, every vertex $v\in V_E$ has a unique color in $N(v)$. Since every vertex $v$ of the set $V_1\cup V_4$ is a pendant vertex, $v$ has a unique color in $N(v)$. Since every vertex $v\in V_2$ has at least $3$ neighbors in $V_3$ which are assigned distinct colors, $v$ has a unique color in $N(v)$. Every vertex of the set $V_3\cup V_5$ is a $2$-vertex and its neighbors are assigned different colors. Thus, every vertex $v\in V_3\cup V_5$ has a unique color in $N(v)$.
			Clearly, the color $k+1$ appears exactly once in $N(a_i)$ for each vertex $a_i\in A$ . Thus, each vertex $a_i\in A$ has a unique color in $N(a_i)$. 
			Thus, every vertex of $G'$ has a unique color in its neighborhood. Hence, $G'$ is PCF $(k+1)$-colorable.
			
			\medskip
			
			Let $G'$ be PCF $(k+1)$-colorable. Let $f$ be a PCF $(k+1)$-coloring of the graph $G'$. Let $f'$ be a coloring of the graph $G$ such that $f'(v) = f(v)$ for each vertex $v\in V(G)$. For each edge $xy\in E(G)$, there exists a $2$-vertex $v$ in $G'$ such that $vx, vy\in E(G')$. Since $f'$ is a PCF coloring, $f(x)\neq f(y)$. Therefore, $f'(x) \neq f'(y)$ for each edge $xy\in E(G)$. Thus, $f'$ is a proper coloring. Since each pair of vertices of $A$ has a common $2$-neighbor, every vertex of $A$ must be assigned distinct colors, say $1, 2, \ldots, k$. Also, every vertex of the set $V_2$ has a common $2$-neighbor with every vertex of the set $A$. Thus, every vertex of the set $V_2$ must receive the color $k+1$. Since every vertex of $V$ has a neighbor in $V_2$, none of the vertex in $V$ is assigned color $k+1$. Thus, $f$ uses at most $k$ colors on the vertices of the set $V$. Therefore, $f$ is a proper $k$-coloring of the graph $G$.
		\end{proof}
		Therefore, for $k\geq 4$, PCF $k$-\textsc{colorability} is NP-complete for perfect elimination bipartite graphs.
	\end{proof}
	
	By Theorems~\ref{np-pebg-thm1}-\ref{np-pebg-thm2}, we have the following theorem.
	
	\begin{theorem}
		For $k\geq 3$,	PCF $k$-\textsc{colorability} is NP-complete for perfect elimination bipartite graphs.
	\end{theorem}
	
	\section{Hardness of approximation}
	
	In this section, we provide a hardness result for approximating the PCF chromatic number. In particular, we show that for every $\varepsilon>0$, the PCF chromatic number cannot be approximated  within a factor of $O(n^{1-\varepsilon})$ for graphs having $n$ vertices, unless P=NP. To achieve this, we need the following result. 
	
	\begin{theorem}[\cite{Zuckerman}]\label{chromaticnumber}
		For any $\varepsilon > 0$, the chromatic number of a graph cannot be approximated within $O(n^{1-\varepsilon})$ unless P=NP.
	\end{theorem}
	
	\begin{theorem}\label{inapproximation}
		For any $\varepsilon>0$, the PCF chromatic number of a graph cannot be approximated within $O(n^{1-\varepsilon})$ unless P=NP.
	\end{theorem}
	\begin{proof}
		Let $G$ be a graph with the vertex set $V(G) = \{v_1, v_2, \ldots, v_n\}$ and the edge set $E(G) = \{e_1, e_2, \ldots, e_m\}$. We construct a graph $G'$ as shown in Figure~\ref{Dually_chordal_npc} from the graph $G$ in polynomial time using the following steps.
		\begin{enumerate}
			\item Add two vertices $x$ and $y$ and make them adjacent to every vertex of $G$.
			\item Add an edge between $x$ and $y$. 
		\end{enumerate}
		Therefore, $G'=(V(G'), E(G'))$ with $V(G') = V \cup \{x, y\}$ and $E(G') = E_A \cup E_X \cup E_Y \cup \{xy\}$, where $E_A = \{v_i v_j\colon~ v_i v_j \in E(G)\}$,  $E_B = \{x v_i:~ i \in [n]\}$, and $E_Y = \{y v_i:~ i \in [n]\}$.
		
		\begin{figure}[h]
			\begin{center}
				\includegraphics[width=0.5\linewidth]{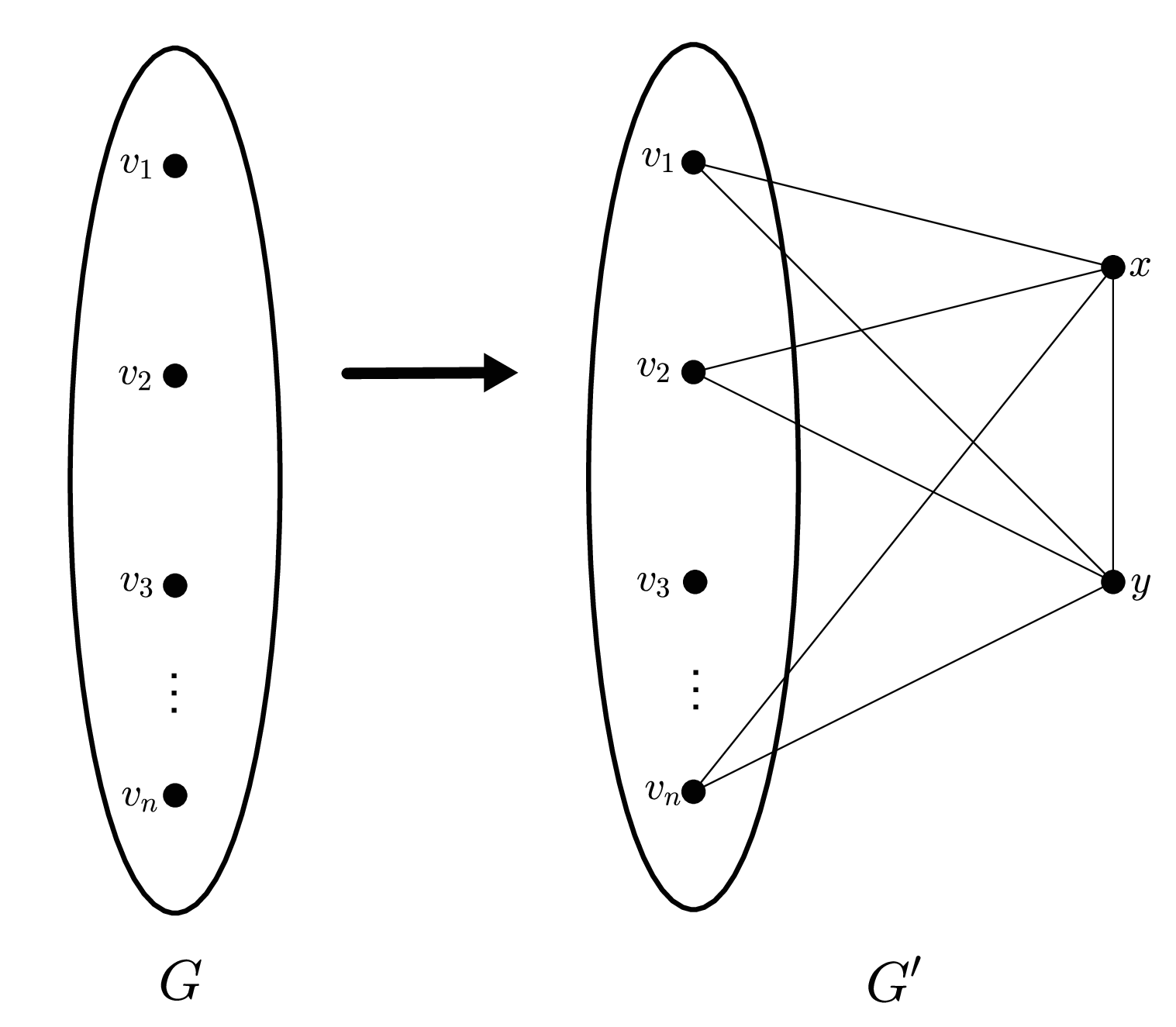}
				\caption{Construction of the graph $G'$ from the graph $G$.}
				\label{Dually_chordal_npc}
			\end{center}	
		\end{figure}
		
		\begin{claim} \label{PCF_dually}
			$G$ is $k$-colorable if and only if $G'$ is PCF $(k+2)$-colorable.
		\end{claim}
		\begin{proof}
			Let $c$ be a proper $k$-coloring of the graph $G$. We extend this coloring to PCF $(k+2)$-coloring of the graph $G'$. Define a coloring $c_* : V(G') \to [k+2]$ for the graph $G'$ such that $c_*(z) = c(z) \;\forall z \in V$, $c^*(x) = k+1$, and $c^*(y) = k+2$. Now, we show that the coloring $c_*$ is a PCF coloring. Since $c$ is a proper coloring and $x$ and $y$ are assigned colors $k+1$ and $k+2$, respectively, $c_*$ is a proper $(k+2)$-coloring. Now, we show that each vertex $v \in V(G')$ has a unique color in $N(v)$. Every vertex $v\in V(G')\setminus\{y\}$ has a unique color $k+2$ in $N(v)$. The vertex $y$ has a unique color $k+1$ in $N(y)$.
			Thus, every vertex of $G'$ has a unique color in its neighborhood. Therefore, $G'$ is PCF $(k+2)$-colorable.
			
			\medskip
			
			Now, let $c_1$ be a PCF $(k+2)$-coloring of the graph $G'$. We show that the graph $G$ is $k$ colorable. We restrict the coloring $c_1$ to the vertices of the set $V$, denoted  by $c_2$. Since $x$ and $y$ are adjacent to all other vertices of $G'$, $c_1(x)$ and $c_1(y)$ does not appear on any other vertex of the graph and $c_1(x) \neq c_1(y)$. Therefore, $c_2$ is a proper coloring of $G$ with $k$ colors. 
		\end{proof}
		
		Given any proper $k$-coloring $f$ of $G$, one can find a PCF coloring $f_{pcf}$ of $G'$ with $k+2$ colors. 
		Now, 
		$$ \frac{|f_{pcf}(V(G'))|}{\chi_{pcf}(G')} =  \frac{|f(V(G))|+2}{\chi(G)+2} \geq \frac{1}{3}\cdot\frac{|f(V(G))|}{\chi(G)}.$$
		Now, assume that there exists a polynomial time algorithm $\mathcal{A}$ that approximates PCF \textsc{colorability} within a factor $N^{1-\varepsilon}$ for a graph with $N$ vertices. Clearly, $N = n+2$. Now, 
		$$ \frac{|f(V(G))|}{\chi(G)} \leq 3\cdot (n+2)^{1-\varepsilon}.$$
		For $n\geq2$, $(n+2)^{1-\varepsilon}\leq 2^{1-\varepsilon}\cdot n^{1-\varepsilon}$.
		So we have,
		$$ \frac{|f(V(G))|}{\chi(G)} \leq 3\cdot2^{1-\varepsilon}\cdot n^{1-\varepsilon} = O(n^{1-\varepsilon}).$$
		This is a contradiction to Theorem~\ref{chromaticnumber}. Therefore, Theorem~\ref{inapproximation} follows. 
	\end{proof}
	
	\medskip
	
	A graph $G$ is called \emph{dually chordal} if it admits a maximum neighborhood ordering, that is, an ordering $(v_1,\dots,v_n)$ of the vertices such that for each $i<n$, $v_i$ has a neighbor $u_i$ in $\{v_{i+1},\dots,v_n\}$ with $N[v_i]\subseteq N[u_i]$ in the subgraph induced by $\{v_i,\dots,v_n\}$. 
	
	\medskip
	
	Now, we show that the graph $G'$ constructed in the proof of Theorem~\ref{inapproximation} is a dually chordal graph. Let $\alpha = (u_1, u_2, \ldots, u_{n-1} = x, u_n = y)$ be an ordering of the vertices of $G'$. 
	Since for every vertex $u_i \in \alpha$ with $i < n$, $u_n$ is a neighbor of $u_i$ in $G_i$ with $N[u_i]\subseteq N[u_n]$, where $G_i = G[V(G)\setminus\{v_1, \ldots, v_{i-1}\}]$. Thus, $G'$ admits a maximum neighborhood ordering. So $G'$ is a dually chordal graph. Therefore, we have following corollary due to Claim~\ref{PCF_dually}.
	\begin{coro}
		PCF $k$-\textsc{colorability} is NP-complete for dually chordal graphs for $k \geq 5$.
	\end{coro}
	
	\section{Polynomial-time algorithms}
	
	In this section, we provide linear-time algorithms for block graphs, proper interval graphs, chain graphs, and pseudo-split graphs. We prove that $\chi_{pcf}(G)\leq\omega(G)+1$ block graphs, proper interval graphs, chain graphs, and pseudo-split graphs (except $C_5$),  and we characterize all graphs for which equality holds.
	
	\medskip
	
	We define some terminologies that will be used in this section.
	A \emph{partial PCF-coloring} of a graph $G$ is defined as a function $f$ that assigns colors to a subset of vertices (not necessarily all) of $G$, say $V'$, such that
	\begin{enumerate}
		\item $f(u)\neq f(v)$ whenever $uv\in E(G)$ and $u, v\in V'$.
		\item Every vertex $v\in V$ with $N(v)\subseteq V'$ has a color that appears exactly once in $N(v)$.
	\end{enumerate}
	
	For a partial coloring $f$ and $M\subseteq V(G)$, we define $f(M)$ to be the set of colors used on the colored vertices of $M$ under $f$.
	
	\subsection{Block graphs}
	
	A vertex $v$ of $G$ is a cut vertex if $G[V\setminus\{v\}]$ is disconnected. A \emph{block} of $G$ is a maximal connected subgraph that has no cut vertex. A graph $G$ is a \emph{block graph} if every block of $G$ is a clique.
	
	\medskip
	
	Note that the block graphs are perfect; thus $\chi(G) = \omega(G)$. Since any PCF coloring of a block graph $G$ requires at least $\chi(G) = \omega(G)$ colors, we first identify block graphs, referred to as \emph{special block graphs} (see Definition \ref{special_block}), for which $\chi_{pcf}(G) \geq \omega(G)+1$. We then present a linear-time algorithm that computes a PCF coloring using $\omega(G)+1$ colors for every special block graph, and $\omega(G)$ colors for every non-special block graph. This establishes a tight upper bound on $\chi_{pcf}(G)$ for block graphs and provides the characterization that $\chi_{pcf}(G) = \omega+1$ if and only if $G$ is a special block graph.
	
	\medskip
	
	Let $G$ be a block graph. Let $\mathcal{B} = \{B_1, B_2, \ldots, B_k\}$ be the set of all the blocks of $G$ and $\mathcal{C} = \{c_1, c_2, \ldots, c_l\}$ be the set of all the cut vertices of $G$ for some integers $k, l\geq 1$. 
	We represent $G$ by a tree-like structure known as \emph{cut-tree} (see \cite{AHU_textbook}), denoted by $T(G)$, where $V(T(G)) = \{B_1, B_2, \ldots, B_k, c_1, c_2, \ldots, \\c_l\}$ and $E(T(G)) = \{(B_i, c_j) :~ c_j \in V(B_i), i \in [k], j \in [l]\}$. The vertices $\{B_1, B_2, \ldots, B_k\}$ are the block vertices of $T(G)$ and the vertices $\{c_1, c_2, \ldots, c_l\}$ are the cut vertices of $T(G)$. The cut-tree $T(G)$ can be obtained in linear-time by a depth-first search of $G$. 
	\begin{figure}[h]
		\begin{center}
			\includegraphics[width=0.9\linewidth]{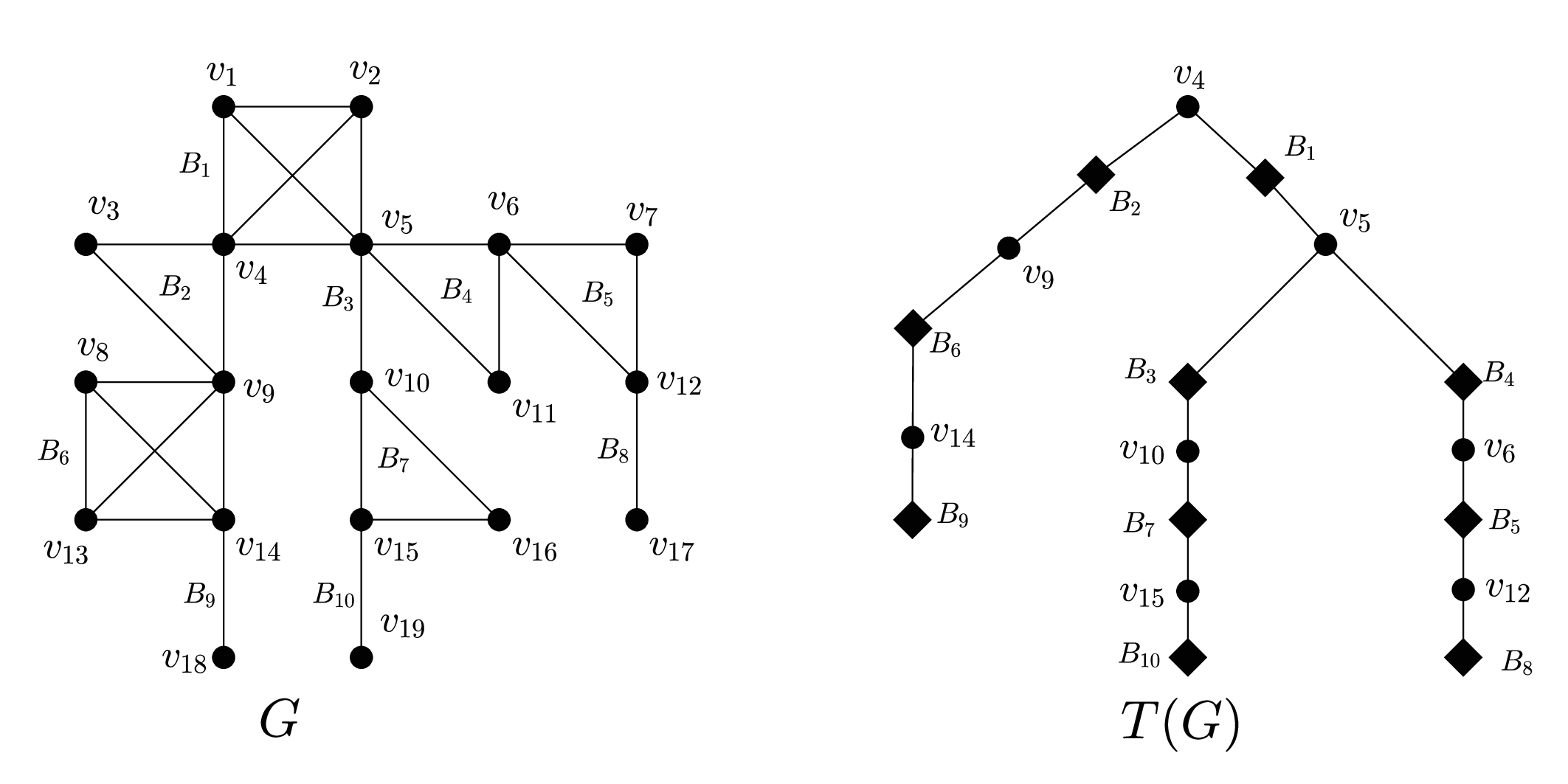}
			\caption{A block graph $G$ and its corresponding cut-tree $T(G)$.}
			\label{cut-tree}
		\end{center}	
	\end{figure}
	
	\begin{defi}[\textsf{Special block graph}]\label{special_block}
		A block graph $G$ is \emph{special} if there exists a cut vertex $v$ such that at least two blocks containing $v$ are maximum cliques of $G$. 
	\end{defi}
	We call a block graph $G$ is said to be \emph{non-special} if $G$ is not special.
	\begin{obs}\label{complete_block_graph}
		If $G$ is a complete graph, then $\chi_{pcf}(G) = \omega(G)$
	\end{obs}
	
	From now on, we assume that $G$ is a block graph that is not complete.
	
	\begin{obs}\label{block_vertex}
		Every vertex $v\in V(G)\setminus \mathcal{C}$ has a unique color in $N(v)$ under any proper coloring of $G$.
	\end{obs}
	\begin{proof}
		Let $v\in V(G)\setminus \mathcal{C}$ be a vertex of $G$. Let $c$ be a proper coloring of the graph $G$. Since $N(v)$ is a clique, every vertex in $N(v)$ receives a distinct color under $c$. Thus, $v$ has a unique color in $N(v)$ under any proper coloring.
	\end{proof}
	
	\begin{lemma}\label{spcl_block_bound}
		If $G$ is a special block graph, then $\chi_{pcf}(G) \geq\omega(G) + 1$.
	\end{lemma}
	
	\begin{proof}
		Let $G$ be a special block graph. Let $\omega(G) = k$. 
		By Definition \ref{special_block}, there exists a cut vertex $v$ such that $v$ is contained in at least two blocks, say $B_1$ and $B_2$, such that $B_1$ and $B_2$ are maximum cliques. For the sake of contradiction, assume that there exists a PCF $k$-coloring $c$ of $G$ with colors $\{1, 2, \ldots, k\}$.
		Without loss of generality, assume that $v$ is assigned color $1$ under $c$. Since $B_1$ is a maximum clique, at least $k-1$ colors different from $1$ must be assigned to the vertices of $B_1$ except $v$. Similarly, at least $k-1$ colors distinct from the color $1$ must be assigned to the vertices of $B_2$ except $v$. Since $v$ is adjacent to every vertex of $B_1$ and $B_2$, every color in the set $\{2, \ldots, k\}$ is used at least twice in $N(v)$. Therefore, $\chi_{pcf}(G) \ge \omega(G) + 1$.
	\end{proof}
	
	Let $T$ be a rooted tree. The \emph{parent} of any vertex $v\in V(T)$ is the immediate predecessor of $v$. Note that every vertex has a unique parent except the root vertex. For a vertex $v\in V(T)$ such that $v$ is not the root of $T$, we say $F(v) = u$ if $u$ is the parent of $v$.
	Let $v$ be a cut vertex of $T(G)$. Let $\mathcal{B}(v)$ be the set of all the blocks containing $v$ and let $\displaystyle V(\mathcal{B}(v)) = \bigcup_{B\in \mathcal{B}(v)}V(B)$.
	
	\begin{defi}[\textsf{Property 1}]\label{property_one}
		A vertex $v$ is said to satisfy \textsf{Property 1} if there exists $B\in \mathcal{B}(v)$ such that $|V(B)| > |V(B')|$ for each $B'\in \mathcal{B}(v)\setminus\{B\}$. 
	\end{defi}
	
	We present our algorithm to compute an optimal PCF coloring of a given block graph $G$. By Observation~\ref{block_vertex}, every vertex $v\in V(G)\setminus \mathcal{C}$ has a color that appears exactly once in $N(v)$ under any proper coloring. Therefore, we color each vertex $u$ of $G$ by avoiding the colors assigned to its neighbors. Moreover, we ensure the existence of a unique color in neighborhood of a cut vertex $v$ if $N(v)\setminus\{u\}$ is already colored.
	
	\medskip
	
	To have a better understanding of the algorithm, we provide a description of the algorithm. Let $f$ be a partial PCF coloring of $G$. 
	Let $f_1:S_1 \to [k]$ and $f_2:S_2 \to [k]$ be such that $S_1, S_2\subseteq V(G)$ and $S_1\cap S_2=\emptyset$ be two partial colorings.
	We define the union of $f_1$ and $f_2$ as $f:S_1\cup S_2\to[k]$ and denote it by $f = f_1\cup f_2$.
	
	\medskip
	
	Let $T(G)$ be a cut-tree of a non-complete block graph $G$ rooted at a cut vertex $c_1$. We start with a partial PCF coloring $f$ of $G$ such that no vertex of $G$ is colored initially.
	We initialize a set of colors $S = \{1\}$ and assign color $1$ to $c_1$. We color blocks of $G$ in breadth first search order of $T(G)$. While coloring each block, we ensure that $f$ remains a partial PCF coloring. 
	Let $v\in \mathcal{C}$ and let $B$ be an uncolored child of $v$ such that every block that precedes $B$ in breadth first search order of $T(G)$ is colored. 
	
	\begin{itemize}
		\item If $|V(B)|\geq |S|$, then we need at least $|V(B)|-|S|$ new colors. Thus, we add $|V(B)|-|S|$ new colors to the set $S$.
		\item We define a set of colors $S_B = f(V(\mathcal{B}(F(B))))\setminus f(F(B))$.
		The role of $S_B$ is to minimize the number of colors used in $N(v)$.
		New colors are used on the vertices of block $B$ only when $S_B$ is exhausted. This ensures that there are at most $\omega+1$ number of colors used in $N(v)$.
		\item Since the parent of each block $B$ is already colored, we need $|V(B)|-1$ colors distinct from $F(B)$ to color all the vertices of $B$. Therefore, we assign a list of $|V(B)|-1$ colors to block $B$ and then color each vertex of $B$ using the procedure \textsc{ColorBlock}($V(B)\setminus F(B), L(B)$).
		\item If $v$ satisfies \textsf{Property 1}, then there exists a block $B'$ containing $v$ such that $|V(B')|\geq |V(B'')|$ for each block $B''\neq B'$ containing $F(B)$.
		Since we color the vertices of each block $B\in \mathcal{B}(v)$ using the colors in $S_B$ first, $v$ has a color that appears exactly once in $N(v)$ under any proper coloring.
		\item Let $B\in \mathcal{B}(v)$ be such that each block $B'\in V(\mathcal{B}(v))\setminus\{B\}$ is colored.
		\item If $v$ does not satisfy \textsf{Property 1}, we color $u$ avoiding all the colors that have appeared on the colored vertices of $N(v)$.
		\item If $G$ is a special block graph, then we avoid at most $\omega(G)$ colors while coloring $u$. Thus, $\omega(G)+1$ colors are sufficient.
		\item If $G$ is not a special block graph, then we avoid at most $\omega(G)-1$ colors while coloring $u$. Thus, $\omega(G)$ colors are sufficient.
		\item Since $f$ remains a partial PCF coloring at every step, $f$ is a PCF coloring at the end of the algorithm.
	\end{itemize}
	
	\begin{procedure}[H]
		\relsize{-1}{ 
			
			\KwIn{A set of vertices $V(B)\setminus F(B)$, a list of colors $L(B)$ of size $|V(B)|-1$}
			\KwOut{A coloring of block vertices of $B$}
			
			\While{$($there exists an uncolored vertex $v \in V(B))$}{
				Let $c \in L(B)$\;
				$f(v) \gets c$\;
				$L(B) \gets L(B) \setminus \{c\}$\;
			}
			\Return{$f$}
			\caption{() \textsc{ColorBlock}($V(B)\setminus F(B), L(B))$}}
	\end{procedure}
	
	\begin{algorithm}[H]\label{block_algo}
		\relsize{-1}{		\DontPrintSemicolon
			\KwIn{A non-complete block graph $G = (V, E)$.}
			\KwOut{A PCF coloring of $G$.}
			Let $T(G)$ be a rooted cut-tree of $G$ rooted at a cut vertex $c_1$;\\
			Traverse the tree $T(G)$ in breadth first search order.\\
			Initialize $S = \{1\}$;\\
			$f(c_1)$ = 1;\\
			\While{$($there exists an uncolored block vertex $B$ in $T(G))$}{   
				Let $S_B = f(V(\mathcal{B}(F(B))))\setminus f(F(B))$.\\
				\If{$(|V(B)| > |S|)$}{
					Let $k = |S|$;\\
					$S = S \cup \{k+1, k+2, \ldots, k+|V(B)|-|S|\}$;\\
				}
				Assign color $f(F(B))$ to the vertex $F(B)$ in block $B$;\\
				\If{$((F(B)$ has an uncolored block-child except $B)$ or $(F(B)$ satisfies \textsf{Property 1}$))$}{
					\If{$(|S_B| \geq |V(B)|-1)$}{
						Let $C'\subseteq S_B$ be obtained by selecting first $|V(B)|-1$ smallest colors from $S_B$;\\
						$f$ = $f\;\cup\;$\textsc{ColorBlock}($V(B)\setminus F(B)$, $C'$);\\   
					} \Else{ 
						Let $C'$ be obtained by selecting first $|V(B)|-1-|S_B|$ smallest colors from $S\setminus (S_{B}\cup\{f(F(B))\})$;\\
						$f$ = $f\;\cup\;$\textsc{ColorBlock}($V(B)\setminus F(B)$, $S_B\cup C'$);\\
					}
				}\Else{  
					\If{$(S\setminus (S_{B}\cup\{f(F(B))\}) = \emptyset)$}{
						Let $k = |S|$;\\
						$S = S\cup \{k+1\}$;\\
					}
					Let $C'\subseteq S_B$ be obtained by selecting first $|V(B)|-2$ smallest colors from $S_B$;\\
					Let $c'\in S\setminus (S_{B}\cup\{f(F(B))\})$ be the least indexed color and let $c\in C'$;\\
					$f$ = $f\;\cup\;$\textsc{ColorBlock}($V(B)\setminus F(B)$, $(C'\setminus\{c\})\cup\{c'\})$;\\
				}
				Assign same colors to all the cut children of $B$ as assigned in $B$;\\           
			}	   
			\Return{f}
			\caption{\textsc{PCF-ColoringBlock($G$)}}}
	\end{algorithm}
	
	\begin{table}[ht]
		\centering
		\small
		\setlength{\tabcolsep}{3pt}
		\renewcommand{\arraystretch}{1.15}
		
		\caption{Illustration of \textsc{PCF-ColoringBlock} on the graph $G$ of Fig.~\ref{cut-tree}}
		\label{block_example}
		
		\begin{tabular}{|p{1.9cm}|c|c|p{1.6cm}|p{2.25cm}|p{4.3cm}|p{2.3cm}|}
			\hline
			
			\textbf{Iteration number} &
			\textbf{Block $B$} &
			\textbf{$F(B)$} &
			~~~~\textbf{ $S$} &
			~~~~~~\textbf{$S_B$} &
			\textbf{Condition} &
			\textbf{Colors used on $V(B)$} \\
			\hline

			0 &  &  & \begin{tabular}[t]{@{}l@{}}
				$~~~~\{1\}$\\
			\end{tabular} &  &  & $f(v_4) = 1$\\
			
			\hline
			
			1 & $B_2$ & $v_4$ & \begin{tabular}[t]{@{}l@{}}
				$~\{1,2,3\}$\\
				
			\end{tabular}
			& $S_{B_2} = \emptyset$ & $\exists$ an uncolored child of $v_4$ except $B_2$ and $v_4$ satisfies \textsf{Property 1} & \begin{tabular}[t]{@{}l@{}}
				$f(v_4)=1$\\
				$f(v_3)=2$\\
				$f(v_9)=3$
			\end{tabular}\\
			
			\hline
			
			2 & $B_1$ & $v_4$ & \begin{tabular}[t]{@{}l@{}}
				$\{1,2,3,4\}$\\
				
			\end{tabular}
			& $S_{B_1} = \{2, 3\}$ & $v_4$ satisfies \textsf{Property 1} & \begin{tabular}[t]{@{}l@{}}
				$f(v_4)=1$\\
				$f(v_1)=2$\\
				$f(v_2)=3$\\
				$f(v_5) = 4$\\
			\end{tabular}\\
			
			\hline
			
			3 & $B_6$ & $v_9$ & \begin{tabular}[t]{@{}l@{}}
				$\{1,2,3,4\}$\\
			\end{tabular} & $S_{B_6} = \{1, 2\}$ & $v_9$ satisfies \textsf{Property 1} & \begin{tabular}[t]{@{}l@{}}
				$f(v_9)=3$\\
				$f(v_8)=1$\\
				$f(v_{13})=2$\\
				$f(v_{14}) = 4$\\
			\end{tabular}\\
			
			\hline
			
			4 & $B_3$ & $v_5$ & \begin{tabular}[t]{@{}l@{}}
				$\{1,2,3,4\}$\\
			\end{tabular} & $S_{B_3} = \{1, 2, 3\}$ & $\exists$ an uncolored child of $v_5$ except $B_3$ and $v_5$ satisfies \textsf{Property 1} & \begin{tabular}[t]{@{}l@{}}
				$f(v_5)=4$\\
				$f(v_{10})=1$\\
			\end{tabular}\\
			
			\hline
			
			5 & $B_4$ & $v_5$ & \begin{tabular}[t]{@{}l@{}}
				$\{1,2,3,4\}$\\
			\end{tabular} & $S_{B_4} = \{1, 2, 3\}$ & $v_5$ satisfies \textsf{Property 1} & \begin{tabular}[t]{@{}l@{}}
				$f(v_5)=4$\\
				$f(v_6)=1$\\
				$f(v_{11})=2$
			\end{tabular}\\
			
			\hline
			
			6 & $B_9$ & $v_{14}$ & \begin{tabular}[t]{@{}l@{}}
				$\{1,2,3,4\}$\\
			\end{tabular} & $S_{B_9} = \{1, 2, 3\}$ & $v_{14}$ satisfies \textsf{Property 1} & \begin{tabular}[t]{@{}l@{}}
				$f(v_{14})=4$\\
				$f(v_{18})=1$\\
			\end{tabular}\\
			
			\hline
			
			7 & $B_7$ & $v_{10}$ & \begin{tabular}[t]{@{}l@{}}
				$\{1,2,3,4\}$\\
			\end{tabular} & $S_{B_7} = \{4\}$ & $v_{10}$ satisfies \textsf{Property 1} & \begin{tabular}[t]{@{}l@{}}
				$f(v_{10})=1$\\
				$f(v_{15})=4$\\
				$f(v_{16})=2$
			\end{tabular}\\
			
			\hline
			
			8 & $B_5$ & $v_{6}$ & \begin{tabular}[t]{@{}l@{}}
				$\{1,2,3,4\}$\\
			\end{tabular} & $S_{B_5} = \{2, 4\}$ & $v_6$ does not satisfy \textsf{Property 1} and $\not\exists$ any uncolored child of $v_6$ except $B_5$ & \begin{tabular}[t]{@{}l@{}}
				$f(v_6)=1$\\
				$f(v_7)=2$\\
				$f(v_{12})=3$
			\end{tabular}\\
			
			\hline
			
			9 & $B_{10}$ & $v_{15}$ & \begin{tabular}[t]{@{}l@{}}
				$\{1,2,3,4\}$\\
			\end{tabular} & $S_{B_{10}} = \{1, 2\}$ & $v_{15}$ satisfies \textsf{Property 1} & \begin{tabular}[t]{@{}l@{}}
				$f(v_{15})=4$\\
				$f(v_{19})=1$\\
			\end{tabular}\\
			
			\hline
			
			10 & $B_8$ & $v_{12}$ & \begin{tabular}[t]{@{}l@{}}
				$\{1,2,3,4\}$\\
			\end{tabular} & $S_{B_8} = \{1, 2\}$ & $v_{12}$ satisfies \textsf{Property 1} & \begin{tabular}[t]{@{}l@{}}
				$f(v_{12})=4$\\
				$f(v_{17})=1$\\
			\end{tabular}\\
			
			\hline
			
		\end{tabular}
	\end{table}
	
	\begin{lemma}\label{cut-vertex_unique}
		Let $c$ be a cut vertex of $G$. Then $c$ has a unique color in $N(c)$ under $f$ returned by the algorithm \textsc{PCF-ColoringBlock($G$)}.
	\end{lemma}
	\begin{proof}
		Suppose the algorithm \textsc{PCF-ColoringBlock($G$)} colors the children of $c$ in the order $B_1, \ldots, B_k$.
		The algorithm \textsc{PCF-ColoringBlock($G$)} colors all the blocks and cut vertices of $T(G)$ preceding $B_k$ in breadth first search order of $T(G)$ before coloring $B_k$. So $c$, every block $B\in \mathcal{B}(c)\setminus\{B_k\}$, and $F(c)$ (if $F(c)$ exists) are colored when the algorithm \textsc{PCF-ColoringBlock($G$)} processes $B_k$. 
		Let $S_{B_k} = f(V(\mathcal{B}(c))\setminus f(c)$. Let $B_{\max}$ be the maximum sized block containing $c$ except $B_k$. If there are more than one such blocks, then pick one block arbitrarily. Note that $f(B) \subseteq f(B_{\max})$ for each block $B\in \mathcal{B}(c)\setminus\{B_k\}$. Also, $S_{B_k} = f(B_{\max})\setminus \{f(c)\}$.
		
		\medskip
		\noindent\textbf{Case 1:} $c$ satisfies \textsf{Property 1}.
		
		\medskip
		The algorithm \textsc{PCF-ColoringBlock($G$)} colors the vertices of $B_k$ by selecting first $|V(B_k)|-1$ smallest colors from the set $S_{B_k}$, when $|V(B_k)|-1 \leq |S_{B_k}|$. Otherwise, the algorithm \textsc{PCF-ColoringBlock($G$)} colors vertices of $B_k$ with the colors of the set $S_{B_k}$ and then colors the remaining vertices with first $|V(B_k)|-S_{B_k}|-1$ smallest colors from the set $S\setminus (S_{B_k}\cup\{f(c)\})$. 
		Since $c$ satisfies \textsf{Property 1}, there exists a block $B'$ containing $c$ such that $|V(B')|\geq |V(B'')|$ for each block $B''\neq B'$ containing $c$. Therefore, there exists a color $c$ that appears on exactly one vertex of exactly one block containing $c$. Thus, $c$ has a unique color in $N(c)$ under $f$. 
		
		\medskip
		\noindent\textbf{Case 2:} $c$ does not satisfy \textsf{Property 1}.
		
		\medskip
		Note that $|V(B_k)|-1 \leq |S_{B_k}|$.
		The algorithm \textsc{PCF-ColoringBlock($G$)} colors the vertices of $B_k$ by selecting $|V(B_k)|-2$ first smallest colors from the set $S_{B_k}$ and one color, say $\alpha$, from the set $S\setminus (S_{B_k}\cup\{f(c)\})$. Since $f(B)\subseteq f(B_{\max})$ for each block $B\in \mathcal{B}(c)\setminus\{B_k\}$ and $S_{B_k} = f(B_{\max})\setminus \{f(c)\}$, $\alpha$ is a unique color in $N(c)$.
	\end{proof}
	
	\begin{lemma}\label{block_correct_time}
		The algorithm \textsc{PCF-ColoringBlock($G$)} returns a PCF coloring of a non-complete block graph $G$ in $O(|V|+|E|)$ time.
	\end{lemma}
	\begin{proof}
		
		The algorithm \textsc{PCF-ColoringBlock($G$)} traverses the cut-tree $T(G)$ in breadth first search order, coloring the vertices of each block vertex $B$ of $T(G)$. 
		The algorithm \textsc{PCF-ColoringBlock($G$)} starts by coloring the root $c_1$ of $T(G)$. 
		Thus, $F(B)$ is already colored whenever an uncolored block $B$ is processed. 
		The algorithm \textsc{PCF-ColoringBlock($G$)} assigns distinct colors to each vertex of $B$ different from $f(F(B))$.  
		Since adjacent vertices receive distinct colors, $f$ is a proper coloring of $G$. Now we show that every vertex $v\in V(G)$ has a unique color $N(v)$. If $v\in V(G)\setminus \mathcal{C}$, then by Observation~\ref{block_vertex}, $v$ has a unique color $N(v)$. If $v\in \mathcal{C}$, then by Lemma~\ref{cut-vertex_unique}, $v$ has a unique color $N(v)$. Therefore, every vertex of $G$ has a unique color in its neighborhood. Thus, $f$ is a PCF coloring of $G$.
		
		\medskip
		
		Now, we analyze the running of the algorithm \textsc{PCF-ColoringBlock($G$)}.
		The cut-tree $T(G)$ can be constructed in $O(|V|+|E|)$ time using a depth-first search \cite{AHU_textbook}. the algorithm \textsc{PCF-ColoringBlock($G$)} traverses each vertex of $T(G)$ in breadth-first order. The algorithm \textsc{PCF-ColoringBlock($G$)} colors vertices of an uncolored block in each iteration. For an uncolored block $B$, the algorithm \textsc{PCF-ColoringBlock($G$)} first calculates $S_B$, then expands the set $S$ (if needed), and then colors the uncolored vertices of $B$ based on some conditions each which can be checked in $O(1)$ time. This can be done in $O(|V(B)|)$ time. Since each vertex is colored exactly once, it takes $O(|V(G)|$ time in processing all blocks of $G$. Therefore, the total running time of the algorithm \textsc{PCF-ColoringBlock($G$)}
		is $O(|V|+|E|)$.
	\end{proof}
	
	\begin{lemma}\label{block_optimal}
		Let $G$ be a non-complete block graph. Then the algorithm \textsc{PCF-ColoringBlock($G$)} uses at most $\ell$ colors, where $\ell =  \begin{cases} 
			\omega(G)+1, & \text{ if } G \text{ is special;} \\
			\omega(G), & \text{ otherwise }.
		\end{cases}$
	\end{lemma}
	\begin{proof}
		By Lemma~\ref{block_correct_time}, the algorithm \textsc{PCF-ColoringBlock($G$)} returns a PCF coloring. We will prove that the algorithm \textsc{PCF-ColoringBlock($G$)} uses at most $\ell$ colors.
		The algorithm \textsc{PCF-ColoringBlock($G$)} initializes a set of available colors $S=\{1\}$. The algorithm \textsc{PCF-ColoringBlock($G$)} enlarges $S$ under one of the following conditions while processing a block $B$:
		
		\begin{enumerate}
			\item $|V(B)|>|S|$.
			
			\item $B$ is the last uncolored block child of $F(B)$ such that $F(B)$ does not satisfy \textsf{Property 1} and $S\setminus (S_{B}\cup\{f(F(B))\}) = \emptyset$.
		\end{enumerate}
		
		If $|V(B)|>|S|$, then the algorithm \textsc{PCF-ColoringBlock($G$)} adds $|V(B)|-|S|$ new colors to $S$ when it encounters a block $B$ with $|V(B)|>|S|$. Since $|V(B)|\leq\omega(G)$ for each block $B$, $|S|\leq\omega(G)$. 
		
		\medskip
		
		Now assume that $B$ is the last uncolored block child of $F(B)$ such that $F(B)$ does not satisfy \textsf{Property 1} and $S\setminus (S_{B}\cup\{f(F(B))\}) = \emptyset$.
		This means that every color in $S$ has already appeared in $N(F(B))$.
		A new color is needed only when $F(B)$ has at least two block children, say $B_1$, $B_2$, such that $|V(B_1)| = |V(B_2)| = |S|$. First, assume that $G$ is not special. By Definition~\ref{special_block}, no cut vertex is contained in two maximum cliques. Therefore, at most one block has size $\omega(G)$ among all blocks containing a cut vertex $v$.
		Therefore, $S$ does not get expanded after $|S|=\omega(G)$.
		Hence, the algorithm \textsc{PCF-ColoringBlock($G$)} uses at most $\omega(G)$ when $G$ is non-special. 
		
		\medskip
		
		Now, assume that $G$ is a special block graph. Thus, there exists a cut vertex $v$ that is contained in at least two blocks $B_1,B_2$ such that $|V(B_1)| = |V(B_2)| = \omega(G)$. When the algorithm \textsc{PCF-ColoringBlock($G$)} processes the last block-child of $v$, it adds one new color to $S$. Since a maximum block containing $v$ has already been processed, we obtain $|S|=\omega(G)+1$. 
		We claim that the algorithm \textsc{PCF-ColoringBlock($G$)} never adds a new color to $S$ after $|S|=\omega+1$. 
		Since the algorithm \textsc{PCF-ColoringBlock($G$)} avoids at most $\omega(G)$ colors while coloring the vertices of a block, at least one color is always available from $S$.
		Hence, the algorithm \textsc{PCF-ColoringBlock($G$)} uses $\omega(G)+1$ colors.
	\end{proof}
	
	\begin{theorem}\label{block_bound_char}
		$\chi_{pcf}(G)\leq\omega(G)+1$. Moreover, $\chi_{pcf}(G)=\omega(G)+1$ if and only if $G$ is a special block graph.
	\end{theorem}
	\begin{proof}
		Since the algorithm \textsc{PCF-ColoringBlock($G$)} returns a PCF coloring of a block graph $G$ with at most $\omega(G)+1$ colors, $\chi_{pcf}(G) \leq \omega(G)+1$. By Claim~\ref{spcl_block_bound}, $\chi_{pcf}(G)\geq \omega(G)+1$ when $G$ is a special block graph. Thus, $\chi_{pcf}(G)=\omega(G)+1$ if and only if $G$ is a special block graph.
	\end{proof}

	\subsection{Proper interval graphs}
	
	Given a family $\mathcal{F}$ of sets, the \emph{intersection graph} of $\mathcal{F}$ is the graph whose vertex set is $\mathcal{F}$, where two vertices $A,B\in\mathcal{F}$ are adjacent if and only if $A\cap B\neq\emptyset$. 
	A graph $G$ is a \emph{proper interval graph} if it is the intersection graph of intervals on the real line such that no interval properly contains another.
	
	\medskip
	
	We note that Sharma et al. \cite{SPP} also studied PCF chromatic number of proper interval graphs. They stated that $\chi_{pcf}(G)\leq \omega(G)+1$, with equality if and only if $\Delta(G)=2\omega(G)-2$. Moreover, they stated that $\chi_{pcf}(G)$ can be computed in linear time. Due to space constraints, proofs of these results were not provided in~\cite{SPP}. In this section, we independently prove these results and provide a linear-time algorithm to compute $\chi_{pcf}(G)$ and construct an optimal PCF coloring of a proper interval graph $G$.

	Since any PCF coloring of $G$ requires at least $\chi(G) = \omega(G)$ colors, we first identify proper interval graphs, referred to as \emph{special proper interval graphs} (see Definition \ref{special_pig}), for which $\chi_{pcf}(G) \geq \omega(G)+1$. We then present a linear-time algorithm that computes a PCF coloring using $\omega(G)+1$ colors for every special proper interval graph, and $\omega(G)$ colors for every non-special proper interval graph. This establishes a tight upper bound on $\chi_{pcf}(G)$ for proper interval graphs and provides the characterization that $\chi_{pcf}(G) = \omega(G)+1$ if and only if $G$ is a special proper interval graph.
	
	\medskip
	
	Let $G$ be a proper interval graph. A vertex $v\in V(G)$ is said to be \emph{simplicial} if $N(v)$ is a clique. Let $\alpha = (v_1, v_2, \ldots, v_n)$ be an ordering of the vertices of the graph $G$. The ordering $\alpha$ is said to be a \emph{perfect elimination ordering} (PEO) if $v_i$ is a simplicial vertex in the graph $G[S]$, where $S = \{v_{i+1}, \ldots, v_n\}$. The ordering $\alpha$ is said to be a \emph{bi-compatible elimination ordering} (BCO) if $\alpha$ and $\alpha^{-1} = (v_n, \ldots, v_1)$ both are perfect elimination orderings. A graph $G$ is a proper interval graph if and only if $G$ has a BCO \cite{REJRL}.
	Let $G$ be a proper interval graph with BCO $\alpha = (v_1, v_2, \ldots,v_n)$. 
	For each $v \in V(G)$, $\min[v] := \min \{\, i : v_i \in N[v] \,\}$
	and $\max[v] := \max \{\, i : v_i \in N[v] \,\}$ with respect to the ordering $\alpha$.

	\begin{defi}[\textsf{Special proper interval graph}]\label{special_pig}
		A proper interval graph $G$ is said to be \emph{special} if there exists a vertex $v\in V(G)$ such that $v$ is contained in two distinct maximum cliques $C_1$ and $C_2$ such that $C_1\cap C_2 = \{v\}$.
	\end{defi}
	
	\begin{figure}[h]
		\includegraphics[width=\textwidth]{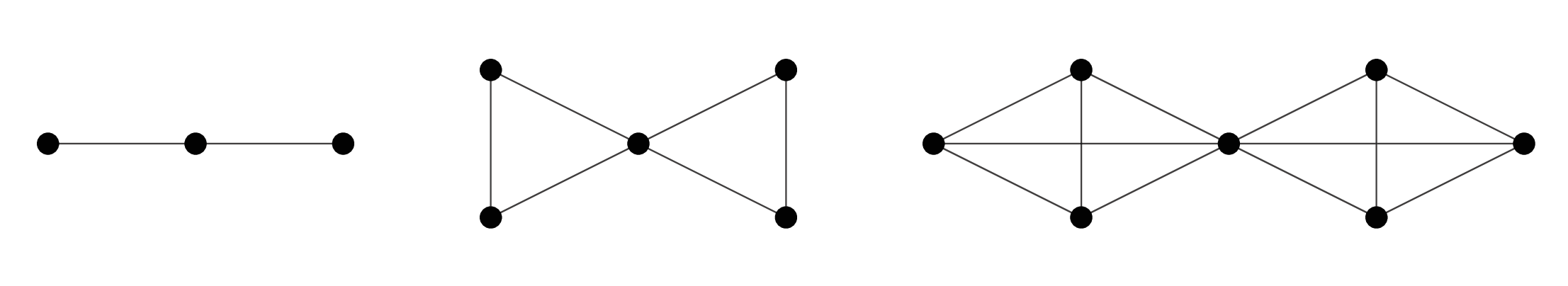}
		\caption{Three special proper interval graphs}%: $\chi_{o}(G) = \omega(G)+1$.}
	\label{special_proper_interval_graph}
\end{figure}

We call a proper interval graph $G$ is said to be \emph{non-special} if $G$ is not special.
\begin{lemma}\label{spcl_pig_bound}
	If $G$ is a special proper interval graph, then $\chi_{pcf}(G) \geq\omega(G)+1$.
\end{lemma}
\begin{proof}
	Let $G$ be a special proper interval graph. Let $\omega(G) = k$. 
	By Definition \ref{special_pig}, there exists a vertex $v$ such that $v$ is contained in two distinct maximum cliques $C_1$ and $C_2$ such that $C_1\cap C_2 = \{v\}$.
	Each of the sets $C_1\setminus\{v\}$ and $C_2\setminus \{v\}$ requires $k-1$ distinct colors in any proper coloring of $G$. Without loss of generality, assign the color $1$ to the vertex $v$.
	Now, every color in the set $\{2, \ldots, k\}$ is used at least twice in $N(v)$. So $v$ has no unique color in $N(v)$ and hence a new color is required in $N(v)$. Thus, $\chi_{pcf} \geq\omega(G)+1$.
\end{proof}

\begin{algorithm}[h]\label{pig_algo}
	\DontPrintSemicolon
	\relsize{-1}{
		\KwIn{A proper interval graph $G = (V, E)$.}
		\KwOut{A PCF coloring of $G$.}
		Let $\alpha = (v_1, v_2, \ldots, v_n)$ be a BCO of $G$.\\
		
		\For{$(i = 1$ to $n)$}{
			Let $c$ be the least indexed color that does not appear on any vertex in $N(v_i)$;\\
			$f(v_i) = c$;\\
			\If{ $(N(v_{\min[v_i]})\subseteq \{v_1, \ldots, v_i\}$ and $v_{\min[v_i]}$ does not have a unique color $N(v_{\min[v_i]})$ $)$}{
				Let $c'\neq c$ be the least indexed color that does not appear on any vertex in $N(v_i)$;\\
				$f(v_i) = c'$;\\
			}
		}
	}
	\Return{f}		
	\caption{\textsc{PCF-ColorPIG}($G$)}
\end{algorithm}

\begin{lemma}\label{special-pig_optimal}
	If $G$ is a proper interval graph, then the algorithm \textsc{PCF-ColorPIG}($G$) returns a PCF coloring $f$ of $G$. 
\end{lemma}
\begin{proof}
	Let $G$ be a proper interval graph. Let $\alpha = (v_1, v_2, \ldots, v_n)$ be a BCO of $G$. For each $i\in [n]$, we prove by induction on iteration $i$ that $f$ is a partial PCF coloring after iteration $i$. Since the algorithm \textsc{PCF-ColorPIG}($G$) assigns color $1$ to the vertex $v_1$,  $f$ is a partial PCF coloring; thus the base case is true. Assume that 
	$f$ is a partial PCF coloring after the iteration $i-1$. 
	Now we show that $f$ is a partial PCF coloring after the iteration $i$.
	
	\medskip
	
	The algorithm \textsc{PCF-ColorPIG}($G$) assigns a color $c$ to $v_i$ avoiding the colors assigned to the vertices of $N(v_i)$. Therefore, no two adjacent vertices are assigned the same color.
	
	\medskip
	
	Let $T = \{v_j:~ N(v_j)\subseteq \{v_1, \ldots, v_{i-1}\}\}$. Since we do not color any vertex in $N(v)$ for any $v\in T$ in the $i^{th}$ iteration, each vertex $v\in T$ has a unique color in $N(v)$ after the $i^{th}$ iteration. Let $M = \{v_j : ~j\leq i, ~\max[v_j]=i\}$. 
	Therefore, it is sufficient to show that each vertex of the set $M$ has a unique color in its neighborhood. If $M = \emptyset$, then we are done. So we may assume that $M \neq \emptyset$. Clearly, $v_{\min[v_i]}\in M$. If $v_{\min[v_i]}$ does not have a unique color in $N(v_{\min[v_i]})$ after coloring $v_i$ with the color $c$, then the algorithm \textsc{PCF-ColorPIG}($G$) recolors the vertex $v_i$ with a new color $c'\neq c$ avoiding the colors of each neighbor of $v_i$. Thus, $v_{\min[v_i]}$ has a unique color in $N(v_{\min[v_i]})$ after $i^{th}$ iteration.
	
	\medskip
	
	Now, let $v_k \in M$ be such that $k\notin \{\min[v_i], i\}$. Let $A = N(v_k)\cap\{v_1, v_2, \ldots, v_k\}$ and $B = N(v_k)\cap\{v_k, v_{k+1} \ldots, v_n\}$. Since $A$ and $B$ are cliques, we have $|A| = |B|$ whenever $v_k$ does not have a unique color in $N(v_k)$ under a proper coloring. Also, $v_k$ does not have a unique color in $N(v_k)$ only when $f(A) = f(B)$. Since $C = \{v_{\min[v_i]}, \ldots, v_i\}$ is a clique and $\min[v_i]< k$, $C$ contains $B$ and at least one vertex of $A$. Thus, we have $f(A) \neq f(B)$. 
	Thus, $v_k$ has a unique color in $N(v_k)$ after coloring $v_{\max[v_i]}$ under $f$. 
	Finally, assume that $v_i\in M$. This is the case only when $i=n$.
	Since $N(v_n)$ is a clique, $v_n$ has a unique color in $N(v_n)$ under $f$.  
	So $f$ is a partial PCF coloring after coloring $v_i$.
	Therefore, by induction $f$ is a partial PCF coloring after iteration $i$ for each $i\in [n]$. Thus, 
	the algorithm \textsc{PCF-ColorPIG}($G$) returns a PCF coloring $f$ of $G$.
\end{proof}

\begin{lemma}
	Let $G$ be a proper interval graph. Then the following are true.
	\begin{enumerate}
		\item If $G$ is not special, then the algorithm \textsc{PCF-ColorPIG}($G$) returns an optimal PCF coloring with $\omega(G)$ colors.
		\item If $G$ is special, then the algorithm \textsc{PCF-ColorPIG}($G$) returns an optimal PCF coloring with $\omega(G)+1$ colors.
	\end{enumerate}
\end{lemma}
\begin{proof}
	(a) Let $v_i\in \alpha$. The algorithm \textsc{PCF-ColorPIG}($G$) avoids at most $\omega(G)-1$ colors while coloring the vertex $v_{\max[v_i]}$. Let $A = N(v_k)\cap\{v_1, v_2, \ldots, v_k\}$ and $B = N(v_k)\cap\{v_k, v_{k+1} \ldots, v_n\}$.
	If $v_i$ does not have a unique color in $N(v_i)$ after assigning a color $c$ to $v_{\max[v_i]}$, then $|A| = |B|$ and $\min[v_{\max[v_i]}] = i$. 
	First, assume that $G$ is non-special. Note that $|B\cup\{v_i\}|\leq\omega(G)-1$. Thus, the algorithm \textsc{PCF-ColorPIG}($G$) avoids at most $\omega(G)-1$ colors while recoloring the vertex $v_{\max[v_i]}$. Since $\chi_{pcf}(G)\geq\omega(G)$ and the algorithm \textsc{PCF-ColorPIG}($G$) returns a PCF coloring with at most $\omega(G)$ colors, the algorithm \textsc{PCF-ColorPIG}($G$) returns an optimal PCF coloring of a non-special graph $G$. 
	
	\medskip
	
	Now, assume that $G$ is special. Note that $|B\cup\{v_i\}|\leq\omega(G)$. Thus, the algorithm \textsc{PCF-ColorPIG}($G$) avoids at most $\omega(G)$ colors while recoloring the vertex $v_{\max[v_i]}$. Since $\chi_{pcf}(G)\geq\omega(G)+1$ by Lemma~\ref{spcl_pig_bound} and the algorithm \textsc{PCF-ColorPIG}($G$) returns a PCF coloring with at most $\omega(G)+1$ colors, the algorithm \textsc{PCF-ColorPIG}($G$) returns an optimal PCF coloring of a special graph $G$. 
\end{proof}

\begin{lemma}
	The algorithm \textsc{PCF-ColorPIG}($G$) computes $f$ in $O(|V|+|E|)$ time.
\end{lemma}
\begin{proof}
	Let $G=(V,E)$ be a proper interval graph and let 
	$\alpha=(v_1,v_2,\ldots,v_n)$ be a BCO $G$.
	A BCO of a proper interval graph can be computed in $O(|V|+|E|)$ time. The algorithm \textsc{PCF-ColorPIG}($G$) colors the vertices in the order $v_1,v_2,\ldots,v_n$.
	Consider the $i^{th}$ iteration.
	The algorithm \textsc{PCF-ColorPIG}($G$) colors the vertex $v_i$ with the smallest indexed color $c$ that does not appear on any vertex of $N(v_i)$. This can be done by checking in $N(v_i)$ and hence it takes $O(\deg(v_i))$ time. After coloring $v_i$, the algorithm \textsc{PCF-ColorPIG}($G$) checks whether $v_{\min[v_i]}$ has a unique color in $N(v_{\min[v_i]})$ when $N(v_{\min[v_i]})\subseteq \{v_1, \ldots, v_i\}$. This can be done by checking each neighbor of $v_{\min[v_i]}$. Thus, it takes $O(\deg(v_{\min[v_i]}))$ time. If recoloring of the vertex $v_i$ is needed, the algorithm \textsc{PCF-ColorPIG}($G$) again finds the smallest indexed color distinct from $c$ that does not appear on any vertex of $N(v_i)$. This takes $O(\deg(v_i))$ time. Therefore, the total running time of all iterations is $\sum_{v \in V} O(\deg(v)) = O(|E|)$.
	Hence, the algorithm \textsc{PCF-ColorPIG}($G$) runs in $O(|V|+|E|)$ time.
\end{proof}

\begin{theorem}\label{pig_bound}
	$\chi_{pcf}(G)\leq\omega(G)+1$. Moreover, $\chi_{pcf}(G)=\omega(G)+1$ if and only if $G$ is a special proper interval graph.
\end{theorem}
\begin{proof}
	Since the algorithm \textsc{PCF-ColorPIG}($G$) returns a PCF coloring of a proper interval graph $G$ with at most $\omega(G)+1$ colors, $\chi_{pcf}(G) \leq \omega(G)+1$.  By Claim~\ref{spcl_pig_bound}, $\chi_{pcf}(G)\geq \omega(G)+1$ when $G$ is a special proper interval graph. Thus, $\chi_{pcf}(G)=\omega(G)+1$ if and only if $G$ is a special proper interval graph.
\end{proof}

\subsection{Chain graphs}

A bipartite graph $G=(X,Y,E)$ is a \emph{chain graph} if there exist orderings $\sigma_X=(x_1,\ldots,x_p)$ of $X$ and $\sigma_Y=(y_1,\ldots,y_q)$ of $Y$ such that $N(x_1)\subseteq N(x_2)\subseteq \cdots\subseteq N(x_p)$ and $N(y_1)\supseteq N(y_2)\supseteq \cdots\supseteq N(y_q)$.
Let $G=(X,Y,E)$ be a chain graph with chain orderings $\sigma_X=(x_1, x_2, \ldots, x_p)$ and $\sigma_Y=(y_1, y_2, \ldots, y_q)$ such that $N(x_1) \subseteq N(x_2) \subseteq \ldots \subseteq N(x_p)$ and $N(y_1) \supseteq N(y_2) \supseteq \ldots \supseteq N(y_q)$. 

\begin{obs}\label{chain_bound}
	$\chi_{pcf}(G)\leq 4$.
\end{obs}

\begin{proof}
	We define a vertex coloring $c$ such that $c(x) = 1$ for each vertex $x\in X\setminus\{x_p\}$, $c(y) = 2$ for each vertex $y\in Y\setminus \{y_1\}$, $c(x_p) = 3$, and $c(y_1) = 4$. We now show that the coloring $c$ is a PCF coloring.  
	Since $c(x) \neq c(y)$ for each edge $e = (x, y) \in E(G)$, $c$ is a proper coloring. Since $N(y_1) = X$ and $c(y_1) = 4$ occurs only once in $Y$, $4$ is a unique color in $N(x)$ for each vertex $x\in X$. Since $N(x_p) = Y$ and $c(x_p) = 3$ occurs only once in $X$, $3$ is a unique color in $N(y)$ for each vertex  $y\in Y$. Thus, every vertex $v \in V(G)$ has a unique color in $N(v)$. Therefore, $c$ is a PCF $4$-coloring. Hence,  $\chi_{pcf}(G)\leq4$.
\end{proof}

\begin{obs}{\label{3pcf-nbd_color}}
	If $c$ is a PCF $3$-coloring of a graph $G$, then every vertex $v\in V(G)$ with $deg(v) \geq 2$ sees exactly two colors in its neighborhood.
\end{obs}

\begin{theorem}\label{chain_characterization_pcf}
	Let $G$ be a chain graph. Then,
	\begin{enumerate}
		\item $\chi_{pcf}(G) = 1$ if and only if $|E(G)|=0$.
		\item $\chi_{pcf}(G) = 2$ if and only if $|E(G)|=1$.
		\item $\chi_{pcf}(G) = 3$ if and only if $|E(G)|\geq 2$ and $G$ has at most two vertices of degree at least $2$.
		\item $\chi_{pcf}(G) = 4$; otherwise.
	\end{enumerate}
\end{theorem}
\begin{proof}
	It is clear that $\chi_{pcf}(G) = 1$ if and only if $|E(G)|=0$ and $\chi_{pcf}(G) = 2$ if and only if $|E(G)|=1$. Therefore, $\chi_{pcf}(G)\geq 3$ if $|E(G)| \geq 2$. 
	Since $|E(G)|\geq 2$, $G$ has a vertex with degree at least $2$. Let chain orderings of $G$ be $\sigma_X = (x_1, x_2, \ldots, x_p)$ and $\sigma_Y = (y_1, y_2, \ldots, y_q)$ such that $N(x_1) \subseteq N(x_2) \subseteq \ldots \subseteq N(x_p)$ and $N(y_1) \supseteq N(y_2) \supseteq \ldots \supseteq N(y_q)$. 
	%First, assume that $G$ does not have any vertex with degree at least $2$. Then, $|E(G)| = 1$ and $\chi_{pcf}(G)=2$.
	First, assume that $G$ has exactly one vertex of degree at least $2$, say $u$. Without loss of generality, assume that $u\in X$.
	Since $N(x_1) \subseteq N(x_2) \subseteq \ldots \subseteq N(x_p)$ and $N(y_1) \supseteq N(y_2) \supseteq \ldots \supseteq N(y_q)$, we have $u=x_p$. Since every vertex of $Y$ is a pendant vertex, $|X| = 1$.
	Define a coloring $c$ such that $c(x_p) = 1$, $c(y) = 2$ for every vertex $y\in Y\setminus \{y_1\}$, and $c(y_1) = 3$. Clearly, $x_p$ has a unique color in $N(x_p)$ assigned to vertex $y_1$. Since every vertex $y\in Y$ is a pendant vertex, $y\in Y$ has a unique color in $N(y)$.
	
	\medskip
	
	Now, assume that $G$ has exactly two vertices of degree at least $2$. Since $N(x_1) \subseteq N(x_2) \subseteq \cdots \subseteq N(x_p)$ and $N(y_1) \supseteq N(y_2) \supseteq \cdots \supseteq N(y_q)$, $x_p$ has the maximum degree among all vertices of $X$ and $y_1$ has the maximum degree among all vertices of $Y$. Therefore, $x_p$ and $y_1$ are the vertices of degree at least $2$. %Let $u = x_p$ and $v = y_1$.
	Define a coloring $c$ such that $c(x_p) = 1$, $c(y_1) = 2$, and $c(v) = 3$ for every vertex $v\in V(G)\setminus\{x_p, y_1\}$. Clearly, $x_p$ has a unique color in $N(x_p)$ assigned to vertex $y_1$. Also, $y_1$ has a unique color in $N(y_1)$ assigned to vertex $x_p$. 
	Since every vertex $v\in V(G)\setminus\{x_p, y_1\}$ is a pendant vertex, every vertex $v\in V(G)\setminus\{x_p, y_1\}$ has a unique color in $N(v)$.
	
	\medskip
	
	Now, assume that $G$ has three vertices, say $u$, $v$, and $w$ of degree at least $2$. Clearly, at least two of these vertices belong to one partition, say $X$. Consequently, $x_p$, $x_{p-1}$, and $y_1$ are vertices with degree at least~$2$. 
	Since $G$ is a chain graph and $x_p$ and $x_{p-1}$ are vertices with degree at least $2$, degrees of $y_1$ and $y_2$ are at least~$2$. Also, the degree of every vertex $x\in N(y_2)$ is at least~$2$. By Observation~\ref{3pcf-nbd_color}, we need two colors, say $c_1$ and $c_2$, in $N(y_2)$. Let $x_j\in X$ be a vertex which is assigned a color different than the color assigned to $x_p$. Since $N(x_j)\subseteq N(x_p)$, $c_1$ and $c_2$ do not appear on any vertex $y\in N(x_j)$. Consequently, we need two new colors distinct from $c_1$ and $c_2$ to color the neighborhood of the vertex $x_j$. Therefore, $\chi_{pcf}(G) \geq 4$. By Observation~\ref{chain_bound}, $\chi_{pcf}(G) \leq 4$. 
	Hence, Theorem~\ref{chain_characterization_pcf} holds.    
\end{proof}

We can find a linear-time algorithm to compute an optimal PCF coloring for chain graphs as shown in the proof of Observation~\ref{chain_bound} and Theorem~\ref{chain_characterization_pcf}.

\subsection{Pseudo-split graphs}

A graph $G$ is \emph{pseudo-split} if $V(G)$ can be partitioned into three (possibly empty) sets $C$, $S$, and $I$ such that 
\begin{enumerate}
	\item $C$ induces a clique, $I$ induces an independent set, and $S$ induces a $C_5$,
	\item every vertex of $C$ is adjacent to every vertex of $S$,
	\item every vertex of $I$ is non-adjacent to every vertex of $S$.
\end{enumerate}

Let $G$ be a pseudo-split graph with $V(G) = S \cup C \cup I$, where $C$ induces a clique, $I$ induces an independent set, and $S$ induces a $C_5$. If $C=\emptyset$ and $S=\emptyset$, then $|V(G)| = \emptyset$. If $C=\emptyset$ and $S\neq\emptyset$, then $G\cong C_5$. Thus, $\chi_{pcf}(G) = 5$. 

\medskip

From now on, we consider $G\ncong C_5$.

\begin{theorem}\label{Sneqemptyset}
	Let $G$ be a pseudo-split graph with $V(G) = S \cup C \cup I$, where $C\neq\emptyset$ induces a clique, $I$ induces an independent set, and $S\neq\emptyset$ induces a $C_5$. Then $\chi_{pcf}(G) = \omega(G)+1$.
\end{theorem}
\begin{proof}
	Let $S = \{s_1, s_2, s_3, s_4, s_5\}$, $C = \{c_1, c_2, \ldots, c_k\}$, and $I = \{x_1, x_2, \ldots, x_r\}$. Since $C$ is a clique with $k$ vertices, at least $k$ colors are needed to color the vertices of $C$. Since every vertex of $C$ is adjacent to every vertex of $S$ and $S$ induces a $C_5$, at least three new colors are needed to color the vertices of the set $S$. %Since $C\cup \{s_1, s_2\}$ forms the maximum clique in $G$, $\omega(G) = k+2$.
	Therefore, $\chi_{pcf}(G)\geq\omega(G)+1$.
	
	\medskip
	
	Let $f$ be a coloring such that $f(c_i) = i$ for each $c_i \in C$, $f(x_i) = k+1$ for each $x_i\in I$, $f(s_1) = f(s_3)=k+1$, $f(s_2) = f(s_4)=k+2$, and $f(s_5) = k+3$. 
	Now, we show that $f$ is a PCF coloring of $G$. Since adjacent vertices are assigned different colors, $f$ is a proper coloring. Since every vertex $v\in C$ is adjacent to every vertex $v'\in S$, $k+3$ is a unique color in $N(v)$ for each vertex $v\in C$. Since every vertex $v'\in S$ is adjacent to every vertex $v\in C$, $f(c_1)$ is a unique color in $N(v')$ for each vertex $v'\in S$. Since $N(v'')$ is a clique for each $v''\in I$, each vertex $v''\in I$ has a unique color in $N(v'')$. Therefore, $f$ is a PCF coloring of the graph $G$. Clearly, $f$ uses $\omega(G)+1$ colors.
	Thus, $\chi_{pcf}(G)=\omega(G)+1$.
\end{proof}

Now, we consider the case when $S$ is empty, that is, $G$ is a split graph. 

\begin{obs}
	If $|C|=1$, then $G$ is a star and $\chi_{pcf}(G) = \begin{cases} 
		2, & \text{ if } G\cong K_2; \\
		3, & \text{ otherwise }.
	\end{cases}$
\end{obs}

From now on, we assume that $|C|\geq 2$.
\begin{defi}[\textsf{Special split graph}]\label{special_split}
	A split graph $G$ is \emph{special} if there exist vertices $v_i, v_j \in C$ such that 
	\begin{enumerate}
		\item if $x\in I$ is not adjacent to $v_i$, then $x$ has a non-neighbor in $C\setminus\{v_i\}$.
		\item if $y\in N(v_i)\cap I$ is not adjacent to $v_j$, then $y$ has a non-neighbor in $C\setminus\{v_j\}$.
	\end{enumerate}
\end{defi}

\begin{theorem}\label{split}
	Let $G$ be a split graph with $V(G) = C \cup I$, where $C$ is a maximal clique with $|C|\geq 2$ and $I$ is an independent set. Then 
	\begin{enumerate}
		\item $\chi_{pcf}(G) = \omega(G)$ if and only if $G$ is special.
		\item $\chi_{pcf}(G) = \omega(G)+1$ if and only if $G$ is not special.
	\end{enumerate}
\end{theorem}
\begin{proof}
	
	To prove (a) and (b), we first prove the following claim.
	\begin{claim}\label{bound_pseudo-split}
		If $G$ is not special, then $\chi_{pcf}(G)\geq \omega(G)+1$.
	\end{claim}
	\begin{proof}
		For the sake of contradiction, assume that there exists a PCF $\omega(G)$-coloring, say $f$.
		Let $G$ be a split graph which is not special. Let $C = \{c_1, c_2, \ldots, c_k\}$. 
		First, assume that $v_i$ does not exist. 
		Consequently, there exists a vertex $x_i\in I$ such that $N(x_i) = C\setminus\{c_i\}$ for each vertex $c_i\in C$. Since $C$ is a clique, each vertex of $C$ must be assigned distinct colors under $f$. Without loss of generality, assume that $f(c_i) = i$ for each $c_i\in C$. Clearly, $x_i$ must be assigned color $i$ for $f$ to be a proper $k$-coloring. Therefore, $f(x_i) = i$ for every vertex $x_i$, $i\in [k]$.
		It is easy to see that none of the vertices of $C$ has a unique color in its neighborhood. This is a contradiction to the fact that there exists a PCF $\omega(G)$-coloring. Thus, $\chi_{pcf}(G)\geq \omega(G)+1$.
		
		\medskip
		
		Now, assume that $v_i$ exists but $v_j$ does not exist.  Since $C$ is a clique, each vertex of $C$ must be assigned distinct colors under $f$. Without loss of generality, assume that $f(c_i) = i$ for each $c_i\in C$. 
		Since $v_j$ does not exist, there exists a vertex $x_j\in N(v_i)\cap I$ such that $N(x) = C\setminus\{v_j\}$ for each vertex $c_j\in C\setminus\{v_i\}$. Note that $x_j$ must be assigned the color $j$ for $f$ to be a proper $k$-coloring. 
		Therefore, each color appears twice in $N(v_i)$. Thus, $v_i$ does not have a unique color in its neighborhood. 
		This is a contradiction to the fact that there exists a PCF $\omega(G)$-coloring. Thus, $\chi_{pcf}(G)\geq \omega(G)+1$.
		This completes the proof of Claim~\ref{bound_pseudo-split}.
	\end{proof}
	
	Now, we prove statements (a) and (b).
	
	\medskip
	
	(a) The necessary part follows by Claim~\ref{bound_pseudo-split}. To prove the sufficiency, let $G$ be a special split graph. Then, either $|C| = 1$ or there exist vertices $v_i, v_j \in C$ such that 
	\begin{itemize}
		\item if $x\in I$ is not adjacent to $v_i$, then $x$ has a non-neighbor in $C\setminus\{v_i\}$.
		\item if $y\in N(v_i)\cap I$ is not adjacent to $v_j$, then $y$ has a non-neighbor in $C\setminus\{v_j\}$.
	\end{itemize}
	
	Let $C = \{c_1, c_2, \ldots, c_k\}$.
	Let $f$ be a coloring such that $f(c_i) = i$ for each $c_i \in C$. Since $C$ is a maximal clique, every vertex of the set $I$ is non-adjacent to some vertex of $C$. Now color every vertex $v\in I\setminus N(v_i)$ with a color from the set $\{1, 2, \ldots, k\}$ different than the colors assigned to its neighbors and $f(v_i)$. Since there is no vertex $x\in I$ such that $f(v_i)$ is the only color available from the set $[k]$ for $x$ due to properness, we can color each vertex $x\in I\setminus N(v_i)$ with a color distinct from its neighbors and $f(v_i)$. 
	Finally, color every vertex $v\in N(v_i)\cap I$ with a color from the set $\{1, 2, \ldots, k\}$ different than the colors assigned to its neighbors and $f(v_j)$. Since there is no vertex $x'\in N(v_i)\cap I$ such that $f(v_j)$ is the only color available from the set $[k]$ for $x'$ due to properness, we can color each vertex $x\in N(v_i)\cap I$ with a color distinct from its neighbors and $f(v_j)$. 
	Clearly, $f$ uses $\omega(G)$ colors. Since every vertex is assigned a color different from the colors assigned to its neighbors, $f$ is a proper coloring. Since $N(v)$ is a clique for each vertex $v\in I$, each vertex $v\in I$ has a unique color in $N(v)$. Since the color $f(v_i)$ does not appear on any vertex of the set $I$, $f(v_i)$ is a unique color in $N(x)$ for each vertex $x\in C\setminus\{v_i\}$. Since the color $f(v_j)$ does not appear on any vertex $v'\in N(v_i)\cap I$, $f(v_j)$ is a unique color of $v_i$ in $N(v_i)$. Since we need at least $\omega(G)$ colors for proper coloring of $G$ and $f$ uses $\omega(G)$ colors, $f$ is an optimal PCF $\omega(G)$-coloring of the graph $G$. Thus, $\chi_{pcf}(G) = \omega(G)$.
	
	\medskip
	
	(b) We provide a PCF coloring with $\omega(G)+1$ colors. Let $f$ be a coloring of $G$ such that $f(c_i) = i$ for every vertex $c_i\in C$ and $f(x) = k+1$ for every vertex $x\in I$. Since adjacent vertices are assigned different colors, $f$ is a proper coloring. Note that $C$ has at least two vertices.
	Clearly, every vertex $c_i\in C\setminus\{c_1\}$ has a unique color $f(c_1)$ in $N(c_i)$. The vertex $c_1$ has a unique color $f(c_2)$ in $N(c_1)$.
	Since $N(v)$ is a clique for each $v\in I$, each vertex $v\in I$ has a unique color in $N(v)$. Therefore, $f$ is a PCF coloring with $\omega(G)+1$ colors. Consequently, $\chi_{pcf}(G)\leq\omega(G)+1$. By Claim~\ref{bound_pseudo-split}, $\chi_{pcf}(G)\geq\omega(G)+1$. 
	Thus, $\chi_{pcf}(G) = \omega(G)+1$.
\end{proof}

Now, we present the running time analysis of finding an optimal PCF coloring for pseudo-split graphs. 
Let $G$ be a pseudo-split graph with $V(G) = S \cup C \cup I$, where $C$ induces a clique, $I$ induces an independent set, and $S$ induces a $C_5$. It can be checked in $O(|V|+|E|)$ time whether $C=\emptyset$. If $C=\emptyset$, then $G\cong C_5$. Therefore, $\chi_{pcf}(G) = 5$ and an optimal coloring can be find in linear-time by assigning distinct colors to each vertex of $G$. Now, we consider the case when $C\neq\emptyset$ and $S\neq\emptyset$. An optimal PCF coloring can be found in linear-time as shown in the proof of Theorem~\ref{Sneqemptyset}. 

\medskip

Finally, we consider the case when $S=\emptyset$. If $|C|=1$ and $G\cong K_2$, then $\chi_{pcf}(G) = 2$. Therefore, we can find an optimal PCF coloring of $G$ in $O(|V|+|E|)$ time. If $|C|=1$ and $G\not\cong K_2$, then $\chi_{pcf}(G) = 3$. We can find an optimal PCF coloring of $G$ in $O(|V|+|E|)$ time. 
First, we see how to check whether $G$ is special. Since $C$ is a maximal clique, $\deg(v)\leq |C|-1$ for each $v\in I$.
First, we construct a set of vertices, say $M$, such that the vertices of the set $C\setminus M$ are the candidates for $v_i$ and $v_j$ as mentioned in Definition~\ref{special_split}. Initialize $M = \emptyset$.
For each vertex $v\in I$, if $\deg(v) = |C|-1$, then add the non-neighbor of $v$ to $M$. This can be done in $O(|E|)$ time.

\begin{obs}\label{first}
	If $v_i$ and $v_j$ exist, then $v_i, v_j\in C\setminus M$.
\end{obs}
\begin{proof}
	First, assume that $v_i\in M$. Consequently, there exists a vertex $x\in I$ such that $x$ does not have a non-neighbor in $C\setminus\{v_i\}$. % is adjacent to each vertex $v\in C\{v_i\}$. 
	This is a contradiction to the definition of $v_i$. Now, assume that $v_i\in C\setminus M$ and $v_j\in M$. Consequently, there exists a vertex $y\in I$ such that $y$ does not have a non-neighbor in $C\setminus\{v_j\}$. %is adjacent to each vertex $v\in C\{v_j\}$. 
	Note that $y\in N(v_i)\cap I$. This is a contradiction to the definition of $v_j$. Therefore, $v_i, v_j\notin M$.
\end{proof}

\begin{obs}\label{second}
	If $|C\setminus M|\geq2$, then $v_i$ and $v_j$ exist.
\end{obs}
\begin{proof}
	Let $u, z\in C\setminus M$. For any $x\in I$, if $x$ is not adjacent to $u$, then $x$ has a non-neighbor in $C\setminus\{u\}$. Therefore, $u$ is a candidate for $v_i$. For any $y\in N(u)\cap I$, if $y$ is not adjacent to $z$, then $y$ has a non-neighbor in $C\setminus\{z\}$. Thus, $z$ is a candidate for $v_j$ corresponding to $v_i$.
	Therefore, $v_i$ and $v_j$ exist.
\end{proof}

If $|C\setminus M| \leq 1$, then $v_i$ and $v_j$ does not exist by Observation~\ref{first} and $G$ is not special. 
If $|C\setminus M| \geq 2$, then by Observation~\ref{second} $v_i$ and $v_j$ exist and $G$ is special. This can be checked in linear-time. Therefore, we can find an optimal PCF coloring of $G$ in $O(|V|+|E|)$ time, as shown in the proof of Theorem~\ref{split}.

\section{Conclusion}
In this paper, we have shown that the PCF chromatic number cannot be approximated within a factor of $O(n^{1-\varepsilon})$, unless P=NP. Using the same reduction, we have shown that PCF $k$-\textsc{colorability} is NP-complete for dually chordal graphs.
On the positive side, we provided polynomial time algorithms in subclasses of dually chordal graphs. We have mainly focused on finding linear-time algorithms. In particular, we present linear-time algorithms to find an optimal PCF coloring of block graphs, proper interval graphs, chain graphs, and pseudo-split graphs.
Moreover, we have shown that $\chi_{pcf}(G)$ is bounded by a linear function $f(\omega(G))$ for all these graph classes. $f(\omega(G))=\omega(G)+2$ for chain graphs; $f(\omega(G))=\omega(G)+1$ for block graphs, proper interval graphs, and pseudo-split graphs (except $C_5$). Further, it is natural to investigate the complexity of PCF $k$-\textsc{colorability} in the broader class of chordal graphs.

\section*{Declarations}

\noindent{\bf Conflict of interest} The authors do not have financial or non-financial interests that are directly or indirectly related to the work submitted for publication.

\noindent{\bf Data availability}
No data was used for the research described in this paper.

\end{document}